\documentclass[11pt]{amsart}
\usepackage{amsmath}
\usepackage{graphicx}
\usepackage{bm}
\usepackage{subfig,tikz}
\usetikzlibrary{decorations.markings}
\usepackage{wrapfig}
\usepackage{xspace}
\usepackage{enumerate}
\theoremstyle{plain}
\newtheorem{theorem}{Theorem}[section]
\newtheorem{lemma}[theorem]{Lemma}
\newtheorem{prop}[theorem]{Proposition}

\theoremstyle{definition}

\numberwithin{equation}{section}

\begin{document}
	

	

\title[Existence and Uniqueness of NHGs]{Existence and Uniqueness of Near-Horizon Geometries for 5-Dimensional Black Holes}
	
\author{Aghil Alaee}
\address{Center of Mathematical Sciences and Applications, Harvard University\\ 20 Garden Street\\
	Cambridge MA 02138, USA}
\email{aghil.alaee@cmsa.fas.harvard.edu}	

		
\author{Marcus Khuri}
\address{Department of Mathematics, Stony Brook University, Stony Brook, NY 11794, USA}
\email{khuri@math.sunysb.edu}
	
\author{Hari Kunduri}
\address{Department of Mathematics and Statistics\\
Memorial University of Newfoundland\\
St John's NL A1C 4P5, Canada}
\email{hkkunduri@mun.ca}
	
	
\thanks{A. Alaee acknowledges the support of NSERC Postdoctoral Fellowship 502873 and support of the John Templeton Foundation. M. Khuri acknowledges the support of NSF Grant DMS-1708798. H. Kunduri acknowledges the support of NSERC Discovery Grant RGPIN-2018-04887.}

\begin{abstract}
We prove existence of all possible bi-axisymmetric near-horizon geometries of 5-dimensional minimal supergravity. These solutions possess the cross-sectional horizon topology $S^3$, $S^1\times S^2$, or $L(p,q)$ and come with prescribed electric charge, two angular momenta, and a dipole charge (in the ring case). Moreover, we establish uniqueness of these solutions up to an isometry of the symmetric space $G_{2(2)}/SO(4)$.
\end{abstract}

	\maketitle
	
\section{Introduction}
An open problem in general relativity is the classification of stationary, asymptotically flat black hole solutions $(M,\mathbf{g})$ in spacetime dimension $D\geq 4$ \cite{chrusciel2012stationary, emparan2008black, hollands2012black}. In the analytic setting, the rigidity theorem asserts that a stationary, rotating black hole is axisymmetric, that is, it admits an additional spatial isometry with closed orbits \cite{hollands2009stationary,Hollands2007,moncrief2008symmetries}. The Killing field generating this isometry commutes with the stationary Killing field. Solutions of the Einstein-Maxwell equations in $D=4$ in this setting can be recast as harmonic maps from the upper half plane to complex hyperbolic space equipped with its canonical metric of negative sectional curvature. This can be used to show that a solution is uniquely specified by three real parameters. Along with the fact that a $D=4$ black hole must have spatial horizon topology $S^2$, this demonstrates that the Kerr-Newman family exhausts all possible analytic, asymptotically flat black hole solutions of the Einstein-Maxwell equations (see \cite{chrusciel2012stationary} for a review).  A non-rotating black hole must be static, and this can be shown to imply it must belong to the Reissner-Nordstr\"{o}m family, which itself is a subset of the Kerr-Newman family.

The analogous classification problem in $D=5$ remains open. The rigidity theorem guarantees only the existence of a single $U(1)$ isometry subgroup. It proves useful to assume the existence of an additional rotational isometry. In this case a harmonic map formulation exists for the pure vacuum Einstein equations and certain supergravity theories. The latter are natural generalizations of standard Einstein-Maxwell theory which are of interest in high-energy physics because their action functionals are invariant under supersymmetry transformations. The topology theorem in this restricted setting asserts that the topology of (a connected component of) the horizon must be $S^3, S^1 \times S^2$, or $L(p,q)$ \cite{galloway2006rigidity,Galloway2006,Hollands2008}.  There is now a large literature on finding explicit solutions in the above class, and in particular, examples have been constructed of several of these horizon types \cite{kunduri2014supersymmetric,Myers1986, Pomeransky2006, tomizawa2016supersymmetric}.  Using the vacuum harmonic map formulation, it has been proved that specification of a certain set of invariants (interval data) uniquely characterize a black hole solution in this class \cite{Figueras2010,Hollands2008}.  More precisely, this data characterizes the fixed point sets of the $U(1)^2$ action on spacetime, as well as the hypersurfaces on which a stationary Killing field is null.  This fixes the topology of the horizon and domain of outer communication,  as well as other geometric invariants such as the mass and angular momenta of the spacetime. Existence, however, remains a challenging open problem although recent progress has been made \cite{Khuri:2017xsc}.

To address the classification problem we will restrict attention to the class of {\it extreme} black holes.  This subset of stationary solutions have vanishing surface gravity, that is the event horizon is a degenerate Killing horizon. Extreme black holes play an important role in various contexts. Physically, for fixed mass they have maximal charge and/or angular momenta and saturate certain geometric inequalities.  In particular, extreme black hole initial data is known to minimize the mass amongst black hole initial data with fixed angular momenta and charge. This has been established in $D=4$ \cite{dain2012geometric} and in $D=5$ for black hole initial data having $S^3$ \cite{Alaee:2015pwa,Alaee:2016jlp} and $S^1 \times S^2$ \cite{Alaee:2017ygv} horizons.  From the standpoint of high-energy physics, the vanishing surface gravity implies that extreme black holes do not Hawking radiate.  Indeed, they are the best understood within theories of quantum gravity such as string theory, which has supplied a statistical account of the Bekenstein-Hawking entropy of a large class of extreme black holes \cite{Strominger1996}.

The key property of extreme black holes that we will exploit is that each admits an associated \emph{near-horizon geometry} \cite{Kunduri:2013ana}. This is a well-defined geometric limit which yields a precise description of spacetime in a neighborhood of the degenerate event horizon.  This limit preserves certain properties of the parent black hole, such as the horizon geometry, although asymptotic information such as the mass and angular velocity is lost. Suitable definitions of charge and angular momenta, however, do exist although there are subtleties in extrapolating these to the analogous quantities defined in the usual way on the asymptotic sphere at spatial infinity.  Classifying near-horizon geometries gives valuable data on the possible set of all extreme black hole. In particular, the absence of a near-horizon geometry with particular horizon topology (or geometry) implies non-existence of an extreme black hole with that horizon. The inverse problem, proving the existence of parent extreme black holes with prescribed asymptotic behavior is a difficult open issue;  progress in this direction has only recently been made in analyzing the moduli space of transverse deformations along the outgoing null radial direction in the axisymmetric case \cite{Li:2015wsa, Li:2018knr} or in the presence of supersymmetry~\cite{Dunajski:2016rtx}.

As we explain below, the near-horizon limit zooms in on the region near the event horizon. This removes a radial degree of freedom.  Thus instead of solving  PDEs on a (stationary) Lorentzian manifold of dimension $D$, the field equations reduce to geometric equations on a $D-2$-dimensional closed Riemannian manifold (a spatial cross-section of the event horizon). This is clearly a considerable simplification, although it is still very difficult in general to perform a classification. Typically an additional assumption must be imposed.  For example, a classification of supersymmetric near-horizon geometries is possible in $D=5$ supergravity \cite{Reall:2002bh}.  In the problem at hand,  $U(1)^{D-3}$-invariant near-horizon geometries of stationary, extreme black holes are cohomogeneity one.  The near-horizon geometries reduce to harmonic maps from a closed interval to a target space with nonpositive sectional curvature.  It turns out that these harmonic maps are singular in the sense that they will blow-up at the endpoints of the orbit space interval.  Given a harmonic map satisfying appropriate boundary conditions,  one can always integrate the rest of Einstein's equations to obtain a near-horizon geometry.

For the large class of gravitational theories which admit a harmonic map formulation (e.g. pure vacuum and various supergravity models) it is possible to integrate the resulting ODEs explicitly to obtain an analytic expression for a coset representative matrix associated with the target space~\cite{kunduri2011constructing}.  In the vacuum, it is possible to determine the harmonic map scalars explicitly and then obtain a full classification of near-horizon geometries of  $D$-dimensional stationary extreme vacuum black holes admitting a $U(1)^{D-3}$ isometry subgroup \cite{hollands2010all}. However, in supergravity theories with more complicated matter content (scalar fields and multiple Maxwell fields) the problem of obtaining the harmonic map functions from the coset matrix reduces to solving a large set of algebraic constraints which is, in practice, not possible.  This makes it impossible to reconstruct the near-horizon geometries from the harmonic map and prevents a classification, mainly because one cannot  implement global regularity conditions on the local solutions.  Thus neither the problem of existence nor uniqueness can be addressed, which is unsatisfactory.

In this work we will instead introduce an abstract approach into the study of this problem of existence and uniqueness for near-horizon geometries in theories admitting a harmonic map formulation. Of particular interest is $D=5$ minimal supergravity, as it arises within the context of string theory compactified on the torus $T^5$.
This approach is originally due to Weinstein \cite{Weinstein}, who exploited the fact that the stationary, axisymmetric Einstein-Maxwell equations in 4-dimensions reduce to a harmonic map with prescribed singularities $\Psi: \mathbb{R}^3\setminus \Gamma \to \mathbb{H}_{\mathbb{C}}^2$ where $\Gamma$ represents the axis of symmetry.  In particular he succeeded in proving the existence and uniqueness of such harmonic maps with prescribed singularities.  This elegant technique yields solutions describing multiple rotating black hole configurations which are smooth away from the axis. Very recently, this method was extended to the $D=5$ bi-axisymmetric vacuum setting by one of the present authors, Weinstein, and Yamada \cite{KhuriWeinsteinYamada1,Khuri:2017xsc}.  This work achieved existence and uniqueness results for singular harmonic maps corresponding to black holes of lens and ring topology and having various asymptotics at infinity.

\section{Statement of Main Results}

We will consider five dimensional spacetimes $(M, \mathbf{g},F)$ where $M$ is a smooth, orientable manifold equipped with a Lorentzian metric $\mathbf{g}$ having signature $(-,+,+,+,+)$, and $F$ is a closed 2-form describing the Maxwell field.  A solution $(M, \mathbf{g},F)$ of $D=5$ minimal supergravity is a critical point of the following action functional
\begin{equation}\label{action}
\mathcal{S} = \int_M R \star 1 - \frac{1}{2} F \wedge \star F - \frac{1}{3\sqrt{3}}  F \wedge F \wedge \mathcal{A},
\end{equation}
where $\star$ is the Hodge dual operator associated to $\mathbf{g}$ and $R$ is scalar curvature. In addition, a local 1-form gauge potential has been introduced so that $F= d\mathcal{A}$, although in general $H_2(M) \neq 0$ so $\mathcal{A}$ need not be globally defined.  This theory automatically includes vacuum general relativity when $F \equiv 0$.  The spacetime field equations derived from this functional are
\begin{equation} \label{SFEintro}
\begin{aligned}
&R_{ab} = \frac{1}{2} F_{ac} F_b^{~c} - \frac{1}{12} |F|^2 \mathbf{g}_{ab} , \\
&d\star F + \frac{1}{\sqrt{3}} F \wedge F =0 .
\end{aligned}
\end{equation}
Unlike the more familiar pure Einstein-Maxwell system, $d \star F \neq 0$.  As discussed above, the theory \eqref{action} is more natural in a variety of contexts. Firstly, it arises in standard dimensional reduction on tori of the 10 and 11-dimensional supergravity theories which govern the low-energy dynamics of string and M-theory.  Secondly, as we discuss below, the field equations reduced on spacetimes having a $U(1) \times U(1)$ action by isometries admit a harmonic map formulation with nonpositively curved target space.

Consider a stationary 5-dimensional spacetime containing a degenerate Killing horizon with Killing field $V$.  This implies that there is an embedded null hypersurface $\mathcal{N}$ on which $|V|^2_\mathcal{N} =0$ and $\nabla_V V|_\mathcal{N} =0$.  A spatial cross section of $\mathcal{N}$ is a 3-dimensional closed Riemannian manifold $\mathcal{H}$.  A spacetime containing an extreme black hole would satisfy these conditions.  In a sufficiently small neighborhood of the event horizon we may always introduce adapted Gaussian null coordinates that describe the near-horizon geometry with spacetime metric
\begin{equation}\label{NHGg}
\mathbf{g}_{\text{NH}} = r^2 \alpha(y) dv^2 + 2dv dr + 2 r \beta_a(y) dv dy^a + \gamma_{ab} dy^a dy^b,
\end{equation}
where $y^a, a = 1\ldots 3$ are local coordinates on $\mathcal{H}$ with $\gamma_{ab}$ its induced Riemannian metric, $\alpha, \beta_a$ are a function and 1-form on $\mathcal{H}$ respectively, and $V = \partial_v$. The event horizon $\mathcal{N}$ corresponds to the null hypersurface $r=0$. Similarly the Maxwell field may be expressed in this coordinate system by
\begin{equation}
F_{\text{NH}} = F_{vr}(y) dv \wedge dr + r F_{va}(y) dv \wedge dy^a + \tilde{F},
\end{equation}
where $\tilde{F}$ is a closed 2-form on $\mathcal{H}$.  In fact the Bianchi identity $dF =0$ further implies that
\begin{equation}\label{NHGF}
F_{\text{NH}} = \sqrt{3} d (\varsigma(y) r dv ) + \tilde{F},
\end{equation}
where we have set $\sqrt{3}\varsigma(y) = - F_{vr}$.  A lengthy computation \cite{Kunduri:2009ud} shows that the spacetime field equations \eqref{SFEintro} are \emph{equivalent} to the following set of equations defined on $\mathcal{H}$
\begin{align}\label{NHGEq1}
\begin{split}
\text{Ric}(\gamma)_{ab} &= \frac{1}{2}\beta_a \beta_b - \nabla_{(a} \beta_{b)} + \frac{1}{2} \tilde{F}_{ac} \tilde{F}_{bd} \gamma^{cd} + \frac{1}{2} \varsigma^2 \gamma_{ab} - \frac{\gamma_{ab}}{12} |\tilde{F}|^2,  \\
\alpha &= \frac{1}{2}\beta_a \beta^a - \frac{1}{2} \nabla_a \beta^a - \varsigma^2 - \frac{|\tilde{F}|^2}{12}, \\
d\star_{\gamma} \tilde{F} &= -\star_\gamma i_\beta \tilde{F} - \sqrt{3} \star_\gamma (d \varsigma - \varsigma \beta) + 2 \varsigma \tilde{F},
\end{split}
\end{align}
where  $\nabla$ is the connection associated to $(\mathcal{H},\gamma)$.  The above formulation of the field equations on near-horizon geometry spacetimes have many advantages over a standard spacetime approach.  In particular, equations \eqref{NHGEq1} are defined on a closed Riemannian manifold $(\mathcal{H},\gamma)$ as opposed to a non-compact Lorentzian one. This is a considerable simplification which facilitates global arguments.

The electric charge associated to the near-horizon geometry is given by
\begin{equation}\label{electriccharge}
\mathcal{Q}  = \frac{1}{16\pi} \int_\mathcal{H} \left( \star F + \frac{1}{\sqrt{3}} \mathcal{A} \wedge F\right).
\end{equation}
Note that the field equations imply that the integrand is a closed 3-form.  If $H_2(\mathcal{H})$ is nontrivial, then a \emph{dipole charge} may be defined by
\begin{equation}\label{dipolecharge}
\mathcal{D}[\mathbf{C}] = \frac{1}{2\pi} \int_\mathbf{C} F,
\end{equation}
for each homology class $[\mathbf{C}] \in H_2(\mathcal{H})$.  This is a `local' charge in the sense that it is not associated to a conserved magnetic charge. Furthermore, by introducing Killing fields $\eta_{(i)}$ that generate the $U(1)^2$ isometry with associated $2\pi$-orbits, so that
\begin{equation}\label{biaxisymmetry}
\mathcal{L}_{\eta_{(i)}} \mathbf{g} = 0 \;,\qquad \mathcal{L}_{\eta_{(i)}} F = 0 \;,
\end{equation}
we may define angular momenta
\begin{equation}\label{angularmomenta}
\mathcal{J}_i = \frac{1}{16\pi} \int_\mathcal{H} \star d \eta_{(i)} + \mathcal{A}(\eta_{(i)}) \left( \star F + \frac{2}{3\sqrt{3}} \mathcal{A} \wedge F\right)
\end{equation}
where the same notation is used to denote the dual 1-form to the rotational Killing fields, and $\mathcal{L}$ represents Lie differentiation.
The field equations and the existence of the isometry imply once again that the integrand is a closed 3-form.

In the presence of a $U(1)^2$ isometry, solutions of the field equations for near-horizon geometries (or equivalently solutions $(\gamma, \beta , \varsigma, \tilde{F})$ of \eqref{NHGEq1}) can be interpreted as singular critical points of a weighted Dirichlet energy for maps from $(-1,1) \to G_{2(2)}/SO(4)$.  The latter is an 8-dimensional non-compact Riemannian manifold equipped with a metric having non-positive sectional curvature.

\begin{theorem}\label{maintheorem}
Given a cross-sectional horizon topology of $S^3$, $S^1\times S^2$, or $L(p,q)$ and values for the electric charge $\mathcal{Q}$, angular momenta $\mathcal{J}_i$, and dipole charge $\mathcal{D}$ (in the ring case), there exists a bi-axisymmetric near-horizon geometry solution of 5-dimensional minimal supergravity realizing these characteristics.
Moreover this solution is unique up to an isometry in the target space $G_{2(2)}/SO(4)$ and a translation in the arc length parameter for the harmonic map.
\end{theorem}

We note that this theorem includes the vacuum case \cite{hollands2010all} in which $\mathcal{Q}=\mathcal{D}=0$, and the target space is replaced by $SL(3,\mathbb{R})/SO(3)$. 
Furthermore, the solutions produced by this result may have conical singularities at the poles, although a proper balancing of angular momentum and charges should alleviate this type of irregularity.

\section{Bi-Axisymmetric Near-Horizon Geometries}

We are interested in proving existence and uniqueness properties of a class of solutions $(\mathbf{g},F)$ of the field equations that describe \emph{near-horizon geometries}.  Such geometries describe in a precise way the spacetime sufficiently close to a degenerate Killing horizon.  The chief motivation for their study is that it is generally expected that any stationary, rotating extreme black hole must contain such a degenerate Killing horizon. In the analytic setting this has been established up to a set of measure zero in the moduli space of solutions \cite{hollands2009stationary}.  This fact is certainly true for all known examples.  Furthermore, each such `parent' extreme black hole has an unique associated near-horizon geometry.  Of course, the existence of a given near-horizon geometry does not guarantee the existence of a black hole spacetime with prescribed asymptotics. Nonetheless, knowledge of the space of near-horizon geometry solutions gives valuable information on the set of allowed extreme black holes.

Let us briefly recall the general notion of a near-horizon geometry before specializing to the $U(1)^2$-invariant setting (for a detailed review, see \cite{Kunduri:2013ana}). Let $\mathcal{N}$ be a Killing horizon with normal Killing vector field $V$.  We may always introduce a Gaussian null coordinate chart $(v,r,y^a)$ in a neighbourhood of $\mathcal{N}$ such that $V = \partial_v$, the horizon is located at $r =0$, and $y^a$ ($a=1,2,3$) are coordinates on $\mathcal{H}$, a spatial (constant $v$) section of $\mathcal{N}$.  It will be assumed that $\mathcal{H}$ is a 3-dimensional compact manifold. In this chart
\begin{align}
\begin{split}
\mathbf{g} &= r^2 \alpha(r,y) dv^2 + 2 dv dr + 2 r \beta_a(r,y) dv dy^a + \gamma_{ab}(r,y) dy^a dy^b, \\
F &= F_{vr} dv \wedge dr + F_{ra} dr \wedge dy^a + F_{va} dv \wedge dy^a + \frac{1}{2} \tilde{F}_{ab} dy^a \wedge dy^b.
\end{split}
\end{align}
Such coordinates are `ingoing' because the radial null vector field $-\partial_r$ is future directed at $r=0$.  The near-horizon geometry is obtained by substituting $v \to v / \varepsilon$, $r \to \varepsilon r$ and letting $\varepsilon \to 0$.  The resulting geometry has metric
\begin{equation}
\mathbf{g}_{\textrm{NH}} = r^2 \alpha(y) dv^2 + 2dv dr + 2 r \beta_a(y) dv dy^a + \gamma_{ab}(y) dy^a dy^b,
\end{equation}
where $(\alpha, \beta, \gamma)$ are defined on $\mathcal{H}$.  The Maxwell field does not automatically admit a well-defined limit; rather, upon use of the Bianchi identity $dF=0$, the identity $\text{Ric}(V,V)|_\mathcal{N} =0$, and the field equation \eqref{SFEintro} one finds $F_{va} =0$ at $r=0$. It then follows (assuming smoothness) that the near-horizon limit of the Maxwell field exists and is given by
\begin{equation}
F_{\textrm{NH}} = -d(F_{vr}(y) r dv) + \tilde{F},
\end{equation}
where $\tilde{F}$ is a closed 2-form on $\mathcal{H}$.  It is convenient to define $\varsigma := -F_{vr}/\sqrt{3}$.  The full spacetime field equations \eqref{SFEintro} for the spacetime fields $(\mathbf{g}_{\text{NH}},F_{\text{NH}})$ are equivalent to the coupled set of equations \eqref{NHGEq1} for the near-horizon data $(\gamma_{ab}, \beta_a, \varsigma, \tilde{F})$ \cite{Kunduri:2013ana}. In particular, the near-horizon spacetime is a solution of the same field equations as its parent extreme black hole.

Suppose now that the spacetime admits a $U(1)^2$ action by isometries, so that \eqref{biaxisymmetry} holds. These isometries extend to the cross-section $(\mathcal{H}, \gamma)$, and the generators of the symmetry are tangent to $\mathcal{H}$. Introduce angular coordinates $\phi^i$, $i=1,2$ associated with these symmetries, having $2\pi$ periodic orbits. Since the interior product of the symmetry generators with the volume form is closed, we may define a function $x$ by
\begin{equation}
dx = \mathcal{C} \text{Vol}_\gamma(\partial_{\phi^1}, \partial_{\phi^2}, \cdot),
\end{equation}
where $\mathcal{C}$ is a constant. As proved in \cite{Hollands2008}, the function $x$ parameterizes the orbit space $\mathcal{H}/ U(1)^2$, and $\mathcal{C}$ may be chosen so that $x \in [-1,1]$.  In the chart $(x,\phi^1, \phi^2)$ the cohomogeneity one horizon metric $\gamma$ then takes the form
\begin{equation}\label{horizonmetric}
\gamma_{ab}dy^a dy^b = \frac{dx^2}{\mathcal{C}^2 \det \lambda} + \lambda_{ij} d\phi^i d\phi^j ,
\end{equation}
and the area of the horizon is
\begin{equation}
A_H = 8\pi^2 \mathcal{C}^{-1}.
\end{equation}
A detailed analysis of the geometry of the torus action can be found in \cite{Hollands2008}. For $x \in (-1,1)$ the torus action is free (the matrix $\lambda_{ij}$ is rank 2), and the endpoints $x=\pm 1$ represent fixed points.  As $x \to \pm 1$, the Killing fields $\mathbf{v}_\pm = a_\pm^i \partial_{\phi^i} \to 0$ where $a^i_\pm \in \mathbb{Z}$.  The matrix $\lambda_{ij}$ is rank 1 at the fixed points, so $\lambda_{ij} a_\pm^i \to 0$ as $x \to \pm 1$.  We are free to choose $a_+ = (1,0)$ and $a_- = (q,p)$ for coprime $p,q \in \mathbb{Z}$ without loss of generality. The topology of $\mathcal{H}$ is then characterized by these integers: $(q,p) = (0,\pm 1)$ corresponds to $S^3$, $(q,p) = (\pm1,0)$ to $S^1 \times S^2$, and otherwise $\mathcal{H} \cong L(p,q)$.

In the remainder of this work we will normalize the area of the horizon by setting $\mathcal{C}=1$, i.e. so that  $A_H = 8\pi^2$.  The vector fields $\mathbf{v}_\pm$ degenerate smoothly at their respective fixed points provided we impose the requirement
\begin{equation}\label{conicalreg}
\lim_{x \to \pm 1} \frac{(1-x^2)^2}{\det \lambda \cdot \lambda_{ij} a_\pm^i a_\pm^j} = 1,
\end{equation}
which eliminates conical singularities. The overall horizon scale can be reinstated by dimensional analysis.

Remarkably, the combination of imposing $U(1)^2$ symmetry along with the field equations \eqref{SFEintro} (or equivalently directly from \eqref{NHGEq1}) on a general near-horizon geometry spacetime $(\mathbf{g}_{\text{NH}}, F_{\text{NH}})$ results in an enhancement of symmetry $\mathbb{R} \times U(1)^2 \to SO(2,1) \times U(1)^2$ \cite{Kunduri:2007vf}. This result holds rather generally for near-horizon geometries in $D$-dimensions invariant under a $U(1)^{D-3}$ torus action satisfying the field equations of general relativity coupled to an arbitrary number of Abelian gauge fields and uncharged scalar fields.  The $SO(2,1)\times U(1)^2$ symmetry constrains solutions to take the form
\begin{align}\label{SO(2,1)met}
\begin{split}
\mathbf{g}_{\text{NH}} =&\Xi(x) \left[-r^2 dv^2 + 2dv dr\right] + \frac{dx^2}{\det \lambda(x)} + \lambda_{ij}(x)\left(d\phi^i + \mathbf{b}^i r dv\right)\left(d\phi^j + \mathbf{b}^j r dv\right), \\
F_{\text{NH}} =&d \left[ \mathbf{a} r dv - \psi^i(x)\left(d \phi^i + \mathbf{b}^i r dv\right)\right].
\end{split}
\end{align}
In the above expressions, the quantities $\mathbf{a}, \mathbf{b}^i$ are constants and $\Xi(x)> 0, \psi^i$ are smooth functions on $\mathcal{H}$.  The 2-dimensional metric in the first square brackets is seen to be the 2-dimensional anti-de Sitter spacetime (AdS$_2$), where the radius of curvature is set to be one.  The action of $SO(2,1)$ yields a torus fibration over AdS$_2$, and it leaves the AdS$_2$ part of the metric invariant. When $\mathbf{b}^i \neq 0$ this action transforms the 1-form $rdv$ by an exact function which can be undone by a corresponding $U(1)$ shift in the direction $\mathbf{b}^i \partial_{\phi^i}$.  Similarly, the Maxwell field $F_{\text{NH}}$ is invariant.


The classification problem for all regular near-horizon geometries is to obtain all solutions $(\mathbf{g}_{NH}, F_{NH})$ to \eqref{NHGEq1}.  In the cohomogeneity one biaxisymmetric setting, using the rigidity result discussed above, one can substitute \eqref{SO(2,1)met} into the near-horizon geometry equations and solve for all possible near-horizon data $(\Xi, \lambda_{ij}, \psi^i, \mathbf{a}, \mathbf{b}^i)$.   For the class of $U(1)^2$-invariant near-horizon solutions we can straightforwardly read off the data appearing in \eqref{NHGg} and \eqref{NHGF}. Namely
\begin{align}
\begin{split}
\alpha =& \frac{1}{\Xi} \left( -1 + \frac{\lambda_{ij} b^i b^j}{\Xi}\right), \qquad \beta  = \frac{\lambda_{ij} b^i}{\Xi} d\phi^j - \frac{\dot \Xi}{\Xi} dx, \\
\varsigma =& \frac{\mathbf{a} - \psi^i \mathbf{b}^i}{\sqrt{3} \Xi} , \qquad \tilde F  = -d\psi^i \wedge d\phi^i ,
\end{split}
\end{align}
where $\dot\Xi$ denotes differentiation with respect to $x$.  From \eqref{NHGEq1} it follows that
\begin{equation}\label{NHG1}
\frac{d}{dx}(  \lambda \dot \Xi) = -2 + \frac{\mathbf{b}^i \mathbf{b}^j \lambda_{ij}}{\Xi}+ \frac{ \lambda \Xi}{3} \dot \psi^i \dot \psi^j \lambda^{ij} + \frac{2}{3\Xi} (\mathbf{a} - \psi^i \mathbf{b}^i)^2.
\end{equation}
Moreover, the $(xx)$ and $(ij)$ components of \eqref{NHGEq1} yield
\begin{align}\label{NHG2}
\begin{split}
-\frac{\ddot \lambda}{2} + \frac{\dot \lambda^2}{4\lambda} - \frac{\lambda}{4}  \lambda^{ik} \dot \lambda_{kj} \lambda^{jl} \dot \lambda_{il} & =-\frac{\lambda \dot \Xi^2}{\Xi^2} + \frac{\dot \lambda \dot \Xi}{2\Xi} + \frac{\lambda \ddot \Xi}{2\Xi} + \frac{\lambda}{3} \dot \psi^i \dot \psi^j \lambda^{ij} + \frac{(\mathbf{a} - \psi^i \mathbf{b}^i)^2}{6\Xi^2},   \\
-\lambda \ddot \lambda_{ij} - \lambda \dot \lambda_{ij} + \lambda \lambda^{kl} \dot \lambda_{ik} \dot \lambda_{jl} &= \lambda_{ik} \lambda_{jl} \mathbf{b}^k \mathbf{b}^l + \lambda \dot\psi^i \dot \psi^j + \lambda_{ij} \left[ \frac{(\mathbf{a} - \psi^k \mathbf{b}^k)^2}{3\Xi} - \frac{\lambda}{3} \lambda^{kl} \dot \psi^i \dot\psi^l + \frac{\lambda \dot \Xi}{\Xi} \right]. 
\end{split}
\end{align}
Finally the near-horizon geometry Maxwell equation reduces to
\begin{equation}
\frac{d}{dx} \left[ \Xi \lambda^{ik} \lambda \dot \psi^i \right] =  (\mathbf{a} - \psi^i \mathbf{b}^i) \left[ \frac{b^k}{\Xi}  - \frac{2}{\sqrt{3} }   \dot\psi^i \epsilon^{ik}\right], \label{NHG4}
\end{equation}
where $\epsilon^{ik}$ is a totally antisymmetric tensor with $\epsilon^{12} = 1$.

Solving the coupled ODEs \eqref{NHG1}-\eqref{NHG4} in full generality by direct means appears to be very difficult. In the vacuum case with vanishing Maxwell fields ($\psi^i, \mathbf{a} \equiv 0$), the direct approach is tractable and led to a full classification \cite{Kunduri2009}.  However, when $F \neq 0$ the situation becomes clearly more difficult and even in the static case (when $\partial_v$ is hypersurface orthogonal) a full classification by direct integrations remains open \cite{Kunduri:2009ud}. Following \cite{kunduri2011constructing} we will exploit the fact that the supergravity field equations \eqref{SFEintro} with $U(1)^2$ symmetry admit a harmonic map formulation from $[-1,1]$ to a non-positively curved symmetric space target. This fact will be our main tool to establish existence and uniqueness of biaxisymmetric solutions of \eqref{NHGEq1}.  In the following subsection we will recall the salient features of this formulation before applying it to spacetimes satisfying \eqref{SO(2,1)met}.

\subsection{Harmonic map formulation of supergravity}
Five-dimensional minimal supergravity may be regarded as the natural extension of standard Einstein-Maxwell theory in that it has a number of important mathematical properties.  In particular, the field equations \eqref{SFEintro} are equivalent to a harmonic map when enough toroidal isometries are present, in this case $U(1)^2$. To see this we summarize the relevant parts of the construction \cite{kunduri2011constructing}. First decompose the spacetime metric as
\begin{equation}\label{lambda1}
\mathbf{g}_{ab} =\frac{1}{\det {\lambda}}h_{ab}+{\lambda}^{ij}\eta_{(i)a}\eta_{(j)b},
\end{equation}
where $h$ is to be regarded as a smooth Lorentzian metric on the 3-dimensional orbit space $M / U(1)^2$.  There are globally defined scalar potentials $\psi^i$ defined by
\begin{equation}\label{psi^i}
d \psi^i = \iota_{\eta_{(i)}} F,
\end{equation}
after utilizing $dF =0$ and topological censorship.  It is straightforward to show that
\begin{equation}
\mathfrak{L}_{\eta_{(i)}}\psi^j = \iota_{\eta_{(i)}} \iota_{\eta_{(j)}} F =0,
\end{equation}
so that the $\psi^i$ are functions defined on the orbit space. We may also define a 1-form
\begin{equation}
\Upsilon = -\iota_{\eta_{(1)}} \iota_{\eta_{(2)}} \star F
\end{equation}
that as a consequence of the Maxwell equation, satisfies
\begin{equation}
d\Upsilon =\frac{1}{\sqrt{3}} d\left(\psi^1 d\psi^2 - \psi^2 d\psi^1 \right).
\end{equation}
It follows that a globally defined electric potential $\chi$ exists and satisfies
\begin{equation}
d\chi = \Upsilon - \frac{1}{\sqrt{3}}\left(\psi^1 d\psi^2 - \psi^2 d\psi^1 \right).
\end{equation}
Next, recall that in pure vacuum the twist potentials $\Theta^i$ are closed 1-forms.  In the supergravity case they are no longer closed, since the Ricci tensor is nonvanishing. Using the field equations, a computation \cite{kunduri2011constructing} shows
\begin{equation}
d\Theta^i=-\Upsilon\wedge \iota_{\eta_{(i)}} F
= d\left[\psi^i \left(d\chi + \frac{1}{3\sqrt{3}}(\psi^1 d\psi^2 - \psi^2 d\psi^1)\right)\right].
\end{equation}
It follows that there exist globally defined \emph{charged twist potentials} $\zeta^i$ which obey
\begin{equation}\label{eq2.11}
d\zeta^i = \Theta^i - \psi^i \left[d\chi + \frac{1}{3\sqrt{3}}(\psi^1 d\psi^2 - \psi^2 d\psi^1)\right].
\end{equation}
These reduce to the vacuum twist potentials when $F \equiv 0$.  Finally, note that the Maxwell field can be reconstructed from the fields $(\lambda_{ij}, \chi, \psi^i, \zeta^i)$ with the identity
\begin{equation}\label{Maxwell}
F = \frac{1}{\det {\lambda}} \left[ \star(\eta_{(2)} \wedge \eta_{(1)} \wedge \Upsilon) + (\det \lambda) {\lambda}^{ij} \eta_{(i)} \wedge d \psi^j \right] .
\end{equation}

A long calculation \cite{Bouchareb:2007ax} now gives the following reformulation of minimal supergravity. Namely, the supergravity field equations \eqref{SFEintro} for $U(1)^2$-invariant solutions $(\mathbf{g},F)$ are equivalent to the following system
\begin{align}
\begin{split}\label{reducedEFE}
\mathrm{Ric}(h)_{ac} =&  \frac{1}{4}\mathrm{Tr}\left({\lambda}^{-1}\nabla_a{\lambda}{\lambda}^{-1}
\nabla_c{\lambda}\right)+\frac{\nabla_a\det{\lambda}\nabla_c
\det{\lambda}}{4(\det{\lambda})^2}\\
&+ \frac{\Upsilon_a \Upsilon_c}{2 \det  \lambda} + \frac{\lambda^{ij}}{2} d\psi_a^i d\psi_c^j + \frac{1}{2 \det \lambda}\lambda^{ij} \Theta_{a}^i \Theta_{c}^j,
\end{split}
\end{align}
\begin{equation}\label{lapchi}
\mathrm{div}_h \left(\frac{\Upsilon}{\det\lambda} \right) = -\frac{(d\psi^i, \lambda^{ij} \Theta^j)_h}{\det\lambda},
\end{equation}
\begin{equation}\label{lappsi}
\mathrm{div}_h (\lambda^{ij} d\psi^j) = \frac{(\Upsilon, \lambda^{ij}\Theta^j)_h}{\det\lambda} - \frac{2}{\sqrt{3}\det\lambda} \left( \delta^{i2} (\Upsilon,d\psi^1)_h - \delta^{i1} (\Upsilon,d\psi^2)_h \right),
\end{equation}
and
\begin{equation}\label{lapzeta}
\mathrm{div}_h \left(\frac{\lambda^{ij} \Theta^j}{\det \lambda}\right) =0,
\end{equation}
where
\begin{align}\label{laplacianlambda}
\begin{split}
\Delta_h {\lambda}_{ij} =& {\lambda}^{kl}(d \lambda_{ik} , d\lambda_{lj})_h - \frac{(\Theta^i ,  \Theta^j)_h}{\det {\lambda}} - (d\psi^i ,d\psi^j)_h +\frac{1}{3} \lambda_{ij}  \left(\lambda^{mn} (d\psi^m , d\psi^n)_h - \frac{(\Upsilon,\Upsilon)_h}{\det \lambda}\right)
\end{split}
\end{align}
and $(\cdot, \cdot)_h$ denotes the inner product on forms with respect to the metric $h$.
Note that the final three equations above are second order elliptic equations for the electromagnetic potentials $(\zeta^i,\chi,\psi^i)$.

Solutions of this system of equations arise as critical points of a 3-dimensional theory of gravity on $(M/U(1)^2,h)$ coupled to a wave map having nonpositively curved symmetric space target $G_{2,2}/SO(4)$, governed by the functional \cite{Bouchareb:2007ax,Possel:2003yw}
\begin{equation}\label{wavemap}
\mathcal{S}[h, X] = \int_{M/U(1)^2} \left(R_h - 2h^{ab} G_{mn} \partial_a X^m \partial_b X^n \right) \; \mathrm{Vol}_h.
\end{equation}
Here $R_h$ is the scalar curvature of $h$, and $X = (\lambda_{ij}, \zeta^i, \chi, \psi^i)$ are coordinates on the target manifold with metric
\begin{align}\label{targetmetric}
\begin{split}
G_{mn} dX^m dX^n =&\frac{(d \det  \lambda)^2}{8 (\det \lambda)^2} + \frac{\mathrm{Tr} ( \lambda^{-1} d \lambda)^2}{8} + \frac{\lambda^{ij} \Theta^i \Theta^j}{4 \det \lambda} + \frac{\Upsilon^2}{4 \det \lambda} + \frac{ \lambda^{ij} d\psi^i d\psi^j}{4} .
\end{split}
\end{align}
The Euler-Lagrange equations of \eqref{wavemap} are given by
\begin{align}
\begin{split}\label{HMeqns}
\mathrm{Ric}(h)_{ab} =&\frac{1}{8} {\text{Tr}} (\mathcal{M}^{-1} \partial_a \mathcal{M} \mathcal{M}^{-1} \partial_b \mathcal{M}) = 2G_{mn}\partial_a X^m \partial_b X^n, \\
\nabla^a (\mathcal{M}^{-1} \partial_a \mathcal{M}) =&0 .
\end{split}
\end{align}
An explicit expression \cite{kunduri2011constructing} for the $7\times 7$ positive definite unimodular coset representative is
\begin{equation}
\label{Mexplicit}
\mathcal{M}= \left(
\begin{array}{ccc}  A & B & \sqrt{2} R \\
 B^t & C & \sqrt{2} T \\
 \sqrt{2} R^t & \sqrt{2} T^t & S \end{array} \right),
\end{equation}
where $A$, $C$ are symmetric $3\times 3$ matrices, $B$ is a $3\times 3$ matrix, $R$, $T$ are $3\times 1$ matrices and $S$ is a scalar. By setting $\chi = \sqrt{3}\mu$, $\psi^i = -\sqrt{3} \nu^i$, and $\lambda = \det \lambda_{ij}$ these submatrices may be expressed as
\begin{align}
\begin{split}
S =& 1 + 2(\nu_k \nu^k + \lambda^{-1}\mu^2), \\
R=& \left( \begin{array}{c} (1+\nu_k \nu^k)\nu_i - \frac{\mu}{\sqrt{\lambda}}\epsilon_{i}^{\phantom{i}k}\nu_k + \frac{\mu}{\lambda}\tilde{\zeta}_i \\
-\frac{\mu}{\lambda} \end{array}\right),  \\
T=& \left( \begin{array}{c} (\lambda^{ij} - \frac{\mu}{\sqrt{\lambda}}\epsilon^{ij})\nu_j \\ \lambda^{kl}\nu_k\tilde{\zeta}_l - \mu[1 + \nu_k \nu^k + \frac{\mu^2}{\lambda} - \frac{\epsilon^{kl}}{\sqrt{\lambda}}\nu_k \tilde{\zeta}_l] \end{array} \right), \\
A =& \left( \begin{array}{cc}
(1+\lambda^{-1} \mu^2)\lambda_{ij} + \lambda^{-1} \tilde{\zeta}_i \tilde{\zeta}_j +(2+ \nu_k \nu^k)\nu_i \nu_j+\frac{\mu}{\sqrt{\lambda}} (\nu_i \nu^k \epsilon_{kj} -\epsilon_{ik}\nu^k \nu_j ) & - \lambda^{-1} \tilde{\zeta}_i \\
-\lambda^{-1}\tilde{\zeta}_j & \lambda^{-1} \end{array} \right) ,  \\
B =& \left(\begin{array}{cc} \nu_i \nu^j  - \frac{\mu}{\sqrt{\lambda}}\epsilon_{i}^{\phantom{j}j} + \frac{1}{\sqrt{\lambda}}\tilde{\zeta}_i \nu_k \epsilon^{kj} & \upsilon_i \\
-\frac{\nu_k \epsilon^{kj}}{\sqrt{\lambda}} & \frac{\mu^2}{\lambda} - \frac{\epsilon^{lm}\nu_l \tilde{\zeta}_m}{\sqrt{\lambda}}
\end{array}\right),
\\
C =& \left( \begin{array}{cc}   (1+\nu_k \nu^k)\lambda^{ij} -\nu^i \nu^j & \tilde{\zeta}^i +(\frac{\mu^2}{\sqrt{\lambda}} - \nu_l \tilde{\zeta}_m \epsilon^{lm})  \epsilon^{ik}\nu_k -\mu \nu^i  \\  \tilde{\zeta}^j +(\frac{\mu^2}{\sqrt{\lambda}} - \nu_l \tilde{\zeta}_m \epsilon^{lm})  \epsilon^{jk}\nu_k -\mu \nu^j   & c
\end{array} \right),
\end{split}
\end{align}
where indices have been raised and lowered with $\lambda_{ij}$ and
\begin{align}
\begin{split}
\tilde{\zeta}_i =&\zeta_i-\mu \nu_i , \\
\upsilon_i =& -\left(1 - \frac{\mu^2}{\lambda}\right)\sqrt{\lambda}\epsilon_{i}^{\phantom{i}k}\nu_k  - (2 + \nu_k \nu^k)\mu \nu_i + \lambda^{kl}\nu_k \tilde{\zeta}_l \nu_i  + \left(-\frac{\mu^2}{\lambda} + \frac{\epsilon^{kl}\nu_k \tilde{\zeta}_l}{\sqrt{\lambda}}\right)\tilde{\zeta}_i  - \frac{\mu}{\sqrt{\lambda}}\epsilon_{ik}\tilde{\zeta}^k ,\\
 c =& \tilde{\zeta}^k\tilde{\zeta}_k - 2\mu \nu^k\tilde{\zeta}_k +\lambda[1+ \nu_k \nu^k+(2+\nu_k \nu^k)\lambda^{-1}\mu^2 +\lambda^{-2}(\mu^2- \nu_l \tilde{\zeta}_m \sqrt{\lambda}\epsilon^{lm} )^2].
\end{split}
\end{align}
The above formulation applies generally to any $U(1)^2$-invariant solution of \eqref{SFEintro}. In the next subsection we will restrict attention to near-horizon geometry members of this class.

\subsection{Near-horizon geometry data as a harmonic map}
The near-horizon geometry $(\mathbf{g}, F)$ given by \eqref{SO(2,1)met} admits a $U(1)^2$ action as isometries by construction. Therefore the entire class of solutions must arise as critical points of the functional \eqref{wavemap}.  The advantage of this latter formulation, as opposed to the equivalent formulation \eqref{NHG1}-\eqref{NHG4} is that it allows for abstract theory to be applied.

To begin, we merely have to relate the near-horizon data to the harmonic map data $X$ and orbit space metric $h$. Observe that the dual 1-forms to the rotational Killing fields $\eta_{(i)}$ are given by (with abuse of notation)
\begin{equation}
\eta_{(i)} = {\lambda}_{ij}(d\phi^j + A^j),
\end{equation}
where $A^j$ are 1-forms on the orbit space. The matrix $\lambda_{ij}$ appearing in \eqref{SO(2,1)met} may then be identified with that defined in \eqref{lambda1}. In addition, the functions $\psi^i$ appearing in \eqref{SO(2,1)met} may be identified with those defined in \eqref{psi^i}. The 3-dimensional orbit space $(M/U(1)^2, h)$ corresponding to \eqref{SO(2,1)met} is a warped product of $[-1,1]$ and AdS$_2$ with
\begin{equation}\label{NHbase}
h_{ab} dx^a dx^b = dx^2 + \Xi(x) \det {\lambda}(x)\left[-r^2 dv^2 + 2dv dr\right].
\end{equation}
The remaining scalar potentials are
\begin{align}
\begin{split} \label{NHGscalars}
\partial_x \chi =& \frac{\mathbf{a} -\psi^i \mathbf{b}^i}{\Xi} - \frac{1}{\sqrt{3}}\left(\psi^1 \partial_x \psi^2 - \psi^2 \partial_x  \psi^1\right), \\
	\partial_x \zeta^i =& \frac{ \lambda_{ij}\mathbf{b}^j - \psi^i ( \mathbf{a} - \psi^j \mathbf{b}^j)}{\Xi} +\frac{2}{3\sqrt{3}} \psi^i  \left(\psi^1 \partial_x\psi^2 - \psi^2 \partial_x\psi^1\right),
\end{split}
\end{align}
and it is clear that $A^j = \mathbf{b}^j r dv$.
	
Conversely, these can be inverted to reconstruct a near-horizon geometry from a given set of harmonic map scalars and three-dimensional orbit space metric $h$ of the near-horizon form above.  Using
\begin{equation}
\mathbf{b}^i =\Xi \lambda^{ij} \left(\partial_x \zeta^j + \psi^j \left[\partial_x \chi + \frac{1}{3\sqrt{3}} \left(\psi^1 \partial_x \psi^2 - \psi^2 \partial_x\psi^1\right)\right] \right),
\end{equation}
the constant $\mathbf{a}$ can be determined from the relations \eqref{NHGscalars}, namely
\begin{equation}
\mathbf{a} = \Xi \left(\partial_x \chi ( 1 + \psi^i \psi^j \lambda^{ij} ) + \psi^i \lambda^{ij} \partial_x \zeta^j + \frac{1}{\sqrt{3}}(\psi^1 \partial_x \psi^2 - \psi^2 \partial_x \psi^1)\left(1 + \frac{\psi^i \psi^j \lambda^{ij}}{3}\right) \right) .
\end{equation}

The biaxisymmetric near-horizon geometry equations \eqref{NHG1}-\eqref{NHG4} are equivalent to \eqref{HMeqns}. In fact, for near-horizon geometries the reduced Einstein equation can be immediately integrated to give an explicit expression for $\Xi(x)$. To see this note that the nonvanishing components of the Ricci tensor of the base space metric $h$ of \eqref{NHbase} are
\begin{align}
\begin{split}
\text{Ric}(h)_{vv} =& \frac{Q''(x) r^2}{2} + r^2 ,\\
\text{Ric}(h)_{vr} =& -\frac{Q''(x)}{2} - 1,\\
\text{Ric}(h)_{xx} =& -2\frac{Q''(x)}{Q(x)} + \left(\frac{Q'(x)}{2Q(x)}\right)^2,
\end{split}
\end{align}
where $Q \equiv \Xi\det \lambda$. Since the coset representative matrix $\mathcal{M}$ is independent of the coordinates $v$ and $r$, the right-hand side of the reduced Einstein equation \eqref{HMeqns} only has a nonvanishing $(x,x)$ component. It then follows from $\text{Ric}(h)_{vv} = 0$ that
\begin{equation}\label{HMQ}
Q(x) =  1-x^2,
\end{equation}
as $Q(\pm 1) =\det\lambda(\pm1)=0$ and $\Xi > 0$.  The requirement $\text{Ric}(h)_{rr}=0$ is automatically satisfied.  The harmonic map equation \eqref{HMeqns} for $\mathcal{M}$ may now be expressed as
\begin{equation}\label{M-ODE}
\partial_x \left[(1-x^2) \mathcal{M}^{-1} \partial_x \mathcal{M}\right] =0 ,
\end{equation}
which can be immediately be integrated to yield $(1-x^2)\mathcal{M}^{-1} \partial_x \mathcal{M} = \mathcal{M}_0$ for some constant matrix $\mathcal{M}_0$.
Furthermore, using \eqref{HMQ} produces
\begin{equation} \label{Ricxx}
\text{Ric}(h)_{xx} = \frac{2}{(1-x^2)^2}
\end{equation}
which in turn yields, upon applying \eqref{HMeqns} the requirement that $\text{Tr}(\mathcal{M}_0^2) = 16$.  We emphasize that the difference between our approach and that developed in \cite{kunduri2011constructing} is that in the latter, the harmonic map equations \eqref{M-ODE} are solved for the matrix $\mathcal{M}(x)$ explicitly, in terms of a large set of integration constants.  The downside of this approach, however, is that given the complicated form of $\mathcal{M}$ \eqref{Mexplicit} it is not possible in practice to extract the harmonic map potentials, and hence the near-horizon geometry, from knowledge of $\mathcal{M}(x)$.  Thus it is not clear how to obtain a classification following this method.  In contrast, we will follow a more abstract approach focussing on properties of functional \eqref{wavemap}.

In summary we have shown that for a $U(1)^2$-invariant near-horizon geometry of the form \eqref{SO(2,1)met}, the field equations \eqref{SFEintro} reduce to solving the harmonic map equation (the second line of \eqref{HMeqns}) for the 8 harmonic map scalars $(\lambda_{ij}, \zeta^i, \chi, \psi^i)$, and then reconstructing the remaining near-horizon data $(\mathbf{a}, \mathbf{b}^i)$.

Before concluding this subsection we note that the electric charge \eqref{electriccharge}, angular momenta \eqref{angularmomenta}, and dipole charge \eqref{dipolecharge} (if $\mathcal{H} \cong S^1 \times S^2$) can be expressed in terms of the boundary values of the harmonic map scalars. We refer to \cite{Alaee:2017pex} for the details of the computation and simply give the results here, namely
\begin{equation}\label{Q}
\mathcal{Q} = \frac{\pi}{4}\int_{-1}^1 d\chi = \frac{\pi}{4}(\chi(1) - \chi(-1)),
\end{equation}
\begin{equation}\label{JJ}
\mathcal{J}_i = \frac{\pi}{4} \int_{-1}^1 d\zeta^i = \frac{\pi}{4}(\zeta^i(1) - \zeta^i(-1)),
\end{equation}
and in the case of a ring horizon the dipole charge is
\begin{equation}\label{DD}
\mathcal{D} = \frac{1}{2\pi}\int_{S^2}F = a^i (\psi^i(-1) - \psi^i(1)),
\end{equation}
where $\mathbf{v} = a^i \eta_{(i)}$ is the Killing field that vanishes at the poles of the $S^2$.

\subsection{Relation to harmonic map energy}
In order to perform the existence and uniqueness argument in the next section, it is advantageous to replace the matrix of scalars $\lambda_{ij}$ with a more convenient set of variables. This may be seen as a reparameterization of the target space.  Firstly, note that for the near-horizon geometries discussed above the metric $h$ is completely determined by \eqref{HMQ}, and hence decouples from the harmonic map in the  the functional \eqref{wavemap}. We may  therefore view it as a functional of the variables $X$ alone defined on the closed interval $[-1,1]$ parameterized by the coordinate $x$. Explicitly, one finds that $\mathcal{S}[h,X] = -2\mathfrak{I}$ where
\begin{equation}\label{S[X]}
\mathfrak{I}= \int_{-1}^1 \left[ (1-x^2) G_{mn} \frac{dX^m}{dx} \frac{ dX^n}{dx} - \frac{1}{1-x^2} \right] dx.
\end{equation}
Note that $\mathfrak{I} = 0$ on critical points as a consequence of the 3-dimensional Einstein equation \eqref{HMeqns} upon using \eqref{Ricxx}.

As explained in detail in \cite{Alaee:2017pex}, we now introduce a convenient reparamterization of the target space.  Firstly note that any $U(1)^2$-invariant horizon geometry must be diffeomorphic to $S^1 \times S^2$ or the lens $L(p,q)$ \cite{Hollands2008} (see the discussion following \eqref{horizonmetric}); here $p$ and $q$ are coprime integers.  In the case of the lens topology, introduce new variables $(U,V,W)$ as follows
\begin{align}
\begin{split}
U=&\frac{1}{4}\log\left(\frac{\det\lambda}{p^2(1-x^2)}\right),\\	 V=&\frac{1}{2}\left(\frac{p^2(1+x)\lambda_{11}}
{(1-x)\left[q^2\lambda_{11}-2q\lambda_{12}+\lambda_{22}\right]}\right),\\	
W=&\sinh^{-1}\left(\frac{\lambda_{12}-q\lambda_{11}}{pe^{2U}\sqrt{1-x^2}}\right),
\end{split}
\end{align}
with inverse transformations
\begin{align}
\begin{split}
\lambda_{11}&=e^{2U+V}(1-x)\cosh W,\quad \lambda_{12}=\frac{e^{2U}}{p}\left(\sqrt{1-x^2}\sinh W-q e^{V}(1-x)\cosh W\right), \\
\lambda_{22}&=\frac{e^{2U}}{p^2}\left(q^2e^{V}(1-x)\cosh W-2q\sqrt{1-x^2}\sinh W+e^{-V}(1+x)\cosh W\right).
\end{split}
\end{align}
Note that the regularity condition \eqref{conicalreg} becomes
\begin{equation}
\lim_{x\to \pm 1} \frac{(1-x^2)^2}{\det {\lambda} \cdot a^i_\pm a^j_\pm {\lambda}_{ij}}= 2p^2 \lim_{x\to \pm 1}e^{-6xU-V} = 1.
\end{equation}
For the ring $S^1 \times S^2$ a different parameterization is needed, namely
\begin{equation}
\lambda_{11}=e^{2U+\bar{V}}\cosh W,\qquad \lambda_{22}=e^{2U-\bar{V}}(1-x^2)\cosh W,\qquad \lambda_{12}=e^{2U}\sqrt{1-x^2}\sinh W,
\end{equation}
where
\begin{equation}
\bar{V}=V+2h_1+h_2,
\end{equation}
with
\begin{equation}
h_1=\frac{1}{4}\log(1-x^2), \qquad  h_2=\frac{1}{2}\log\left(\frac{1-x}{1+x}\right).
\end{equation}
The conical singularity regularity condition \eqref{conicalreg} is now expressed as
\begin{equation}
\lim_{x\to \pm 1}e^{-6U+\bar{V}}=\frac{1}{2}.
\end{equation}

We may now rewrite the functional $\mathfrak{I}$ in a unified way for all admissible topologies by associating each with integers $(p,q,s)$, where $s=0,1$, according to the rule
\begin{equation}
\begin{cases}
H \cong S^3,& s=0,\quad p=1,\quad q=0,\\
H \cong L(p,q),& s=0,\quad 1\leq q\leq p-1,\\
H \cong S^1\times S^2,& s=1,\quad p=1,\quad q=0.\\
\end{cases}
\end{equation}
By setting
\begin{equation}
V_s=V+2sh_1+sh_2,
\end{equation}
an involved computation \cite{Alaee:2017pex} reveals that
\begin{equation}\label{RELATION}
\mathfrak{J}=\frac{\mathcal{I}}{4} - \log 2p^2 -\frac{3+s}{2}
\end{equation}
where, upon setting $ x = \cos \theta$ for later convenience
\begin{align}\label{AFunctional}
\begin{split}
&\mathcal{I}(\Psi)\\
=&\int_{0}^{\pi}
\Bigg\{12\left(\partial_{\theta} U\right)^2+(\partial_{\theta} V_s)^2+(\partial_{\theta} W)^2+\sinh^2 W(\partial_{\theta} V+\partial_{\theta}  h_2)^2\\
&+p^2\frac{e^{-6h_{1}-h_2-6U-V}}{\cosh W}(\bar{\Theta}^{1}_{\theta})^{2}+
p^2e^{-6h_{1}+h_2-6U+V}\cosh W
\left(e^{-h_{2}-V}\tanh W\bar{\Theta}^{1}_{\theta}-\bar{\Theta}^{2}_{\theta}\right)^{2}\\
&+ p^2\frac{e^{-2h_1-h_2-2U-V}}{\cosh W}(\partial_{\theta}\bar{\psi}^{1})^{2}+p^2 e^{-2h_{1}+h_2-2U+V}\cosh W(e^{-h_2-V}\tanh W \partial_{\theta}\bar{\psi}^{1}-\partial_{\theta}\bar{\psi}^{2})^{2}\\
&+p^2 e^{-4h_{1}-4U}\Upsilon_{\theta}^{2} -\left[2sV_s \sin\theta
-12U\sin\theta\right]\partial_{\theta}
h_2\Bigg\} \sin\theta d\theta.
\end{split}
\end{align}
Here $\Psi=(U,V,W,\zeta^{1},\zeta^{2},\chi,\psi^{1},\psi^{2})$ and
\begin{equation}
\bar{\Theta}=Z^t \Theta,\quad\quad\quad \bar{\psi}=Z^t \psi,\quad\quad\quad
Z=\left(\begin{array}{cc}
1 & p\\
0 & q
\end{array}
\right).
\end{equation}
It is clear that  $\mathcal{I}$ is finite, and we will now demonstrate that it may be interpreted as a \emph{reduced energy}, that is a renormalization of a singular Dirichlet energy for maps from $S^3 \to G_{2,2}/SO(4)$.

Consider the round metric on $S^3$ given in Hopf coordinates $(\theta, \phi^1, \phi^2)$, where $\theta \in (0,\pi)$, by
\begin{equation}
g_{S^3} = \frac{d\theta^2}{4} + \sin^2 (\theta /2) (d\phi^1)^2 + \cos^2 (\theta/2) (d\phi^2)^2.
\end{equation}
We are interested in harmonic maps from $(S^3, g_{S^3})$ to the symmetric space $G_{2(2)}/SO(4) \cong \mathbb{R}^8$ equipped with the following complete Riemanninan metric of non-positive curvature
\begin{align}
\begin{split}
G=&12du^{2}
+\cosh^{2}w dv^{2}
+dw^{2}+p^2\frac{e^{-6u-v}}{\cosh w}(\bar{\Theta}^{1})^{2}
+p^2e^{-6u+v}\cosh w (e^{-v}\tanh w\bar{\Theta}^{1}-\bar{\Theta}^{2})^{2}\\
&+p^2\frac{e^{-2u-v}}{\cosh w}(d\bar{\psi}^{1})^{2}
+p^2e^{-2u+v}\cosh w(e^{-v}\tanh w d\bar{\psi}^{1}-d\bar{\psi}^{2})^{2}
+p^2e^{-4u}\Upsilon^{2}.
\end{split}
\end{align}
Let $\Omega \subset S^3$ be a domain that avoids $\Gamma$, the union of the two circles at $\theta=0,\pi$, and let $\tilde{\Psi} = (u,v,w, \zeta^1, \zeta^2, \chi, \psi^1, \psi^2):S^3\setminus\Gamma \to G_{2(2)}/SO(4)$ be a singular bi-axisymmetric map. Then the Dirichlet energy on this domain is given by
\begin{align}\label{DEnergy}
\begin{split}
E_{\Omega}(\tilde{\Psi})= &\frac{4}{\pi^2}\int_{\Omega}\Bigg\{12(\partial_{\theta} u)^{2}
+\cosh^{2}w(\partial_{\theta} v)^{2}
+(\partial_{\theta}w)^{2}+p^2\frac{e^{-6u-v}}{\cosh w}(\bar{\Theta}^{1}_{\theta})^{2}\\
&
+p^2e^{-6u+v}\cosh w \left(e^{-v}\tanh w\bar{\Theta}_{\theta}^{1}-\bar{\Theta}_{\theta}^{2}\right)^{2}
+ p^2\frac{e^{-2u-v}}{\cosh w}(\partial_{\theta}\bar{\psi}^{1})^{2}\\
&+p^2 e^{-2u+v}\cosh w\left(e^{-v}\tanh w\partial_{\theta}\bar{\psi}^{1}-\partial_{\theta}\bar{\psi}^{2}\right)^{2}
+p^2 e^{-4u}\Upsilon_{\theta}^{2} \Bigg\} d\mathcal{V},
\end{split}
\end{align}
where $d\mathcal{V}$ is the volume element on $S^3$.

The difference between the renormalized map $\Psi$ and the unrenormalized map $\tilde{\Psi}$ only appears in the first two variables
\begin{equation}
u = h_1 + U, \qquad v = h_2 + V, \qquad w = W,
\end{equation}
where in the new coordinate $h_1 = \tfrac{1}{2} \log \sin \theta$ and $h_2 = \log \tan \tfrac{\theta}{2}$.  Through integration by parts and with the help of $\partial_{\theta}\left(\sin\theta\partial_{\theta}h_2\right)=0$, the functional $\mathcal{I}$ is shown to be related to the harmonic energy via the formula
\begin{align}\label{relationIE}
\begin{split}		 {\mathcal{I}}_{\Omega}(\Psi)=& E_{\Omega}(\tilde{\Psi})-\int_{\Omega}\left(\left(2s\cos^2\frac{\theta}{2}-1\right)^2
+3\cos^2\theta\right)(\partial_{\theta}h_2)^2 d\mathcal{V}\\
			 &+\int_{\partial\Omega}\left(2\left(2s\cos^2\frac{\theta}{2}-1\right)V_s -12\cos\theta U\right)\partial_{\nu}h_2 dA,
\end{split}
\end{align}
where $\nu$ is the unit outer normal. From this it follows that the two functionals share the same critical points.

\section{Existence and Uniqueness of Singular Harmonic Maps}

In this section we prove existence and uniqueness of the relevant harmonic maps with prescribed singularities at the north and south pole circles of $S^3$. Our approach is based on that of Weinstein \cite{Weinstein}, who treated a similar problem for maps from a compact manifold with nonempty boundary into rank one symmetric space targets. Here, however, the setting is more difficult since the domain has no boundary and the target space $G_{2(2)}/SO(4)$ is of rank two. On the other hand our domain has a cohomogeneity one metric, where as in \cite{Weinstein} it is cohomogeneity two.

\subsection{The model map}
Asymptotics for the singular harmonic map, as well as prescribed angular momentum and charges, are encoded in the model map which may be thought of as an approximate solution. The renormalized version of this bi-axisymmetric map will be denoted by $\Psi_0=(U_0,V_0,W_0,\zeta^i_0,\chi_0,\psi^i_0)$. Let $\varepsilon>0$ and set $\Omega_{\varepsilon}=\{\theta\mid |\sin\theta|>\varepsilon\}\times T^2$. On $S^3\setminus\Omega_{\varepsilon}$ we may define $\Psi_0$ to be any smooth map, which interpolates between its prescription near the poles. Near each pole the model map will be set to an exact solution. At the north pole with rod structure $(1,0)$, we may use the extreme charged Myers-Perry near-horizon geometry with potentials arbitrarily prescribed at the pole. At the south pole the rod structure is $(q,p)$, and we may apply an isometry in the target to transform the extreme charged Myers-Perry solution with potentials vanishing at the pole to a solution having this rod structure and again vanishing potentials at the pole. In this way the model map is smooth, satisfies the near-horizon geometry equations near each pole, and yields arbitrarily prescribed angular momentum and charges from the formulae \eqref{Q}, \eqref{JJ}, and \eqref{DD}. Moreover the following properties are immediately implied by the construction.

\begin{lemma}
The reduced energy of the model map $\tilde{\Psi}_0$ is finite and the tension  $\tau(\tilde{\Psi}_0)$ is pointwise bounded.
\end{lemma}

Asymptotics for the model map may be derived from explicitly known near-horizon geometries arising from extreme black holes such as the dipole charged black ring \cite{Kunduri:2007vf}, and black lenses \cite{kunduri2014supersymmetric,tomizawa2016supersymmetric}, namely
\begin{equation}\label{eq4.11}
U_0,\zeta^{1}_0,\zeta^{2}_0,\chi_0=O(1),\quad W_0=O(\sin\theta),\quad \partial_{\theta} U_0,\partial_{\theta}\chi_0,\partial_{\theta}\psi^i_0=O(\sin\theta),\quad \partial_{\theta} W_0= O(1),
\end{equation}
\begin{equation}\label{eq4.12}
V_0=\begin{cases}
O(1)& s=0\\
-2\log\left(\sin\frac{\theta}{2}\right)+O(1) &s=1
\end{cases},\quad \partial_{\theta}V_0=\begin{cases}
O(\sin\theta)& s=0\\
-\cot\frac{\theta}{2}+O(\sin\theta) &s=1
\end{cases},
\end{equation}
\begin{equation}
\psi^1_0=\begin{cases}
O(\sin^2\frac{\theta}{2}) & s=0\\
O(1) & s=1
\end{cases},\qquad \psi^2_0=\begin{cases}
O(\cos^2\frac{\theta}{2}) & s=0\\
O(1) & s=1
\end{cases},\qquad
\Theta^2_0=O(\sin^2\theta), \quad s=1,
\end{equation}
\begin{equation}
\partial_{\theta}\zeta^{1}_0=\begin{cases}
\sin^2\frac{\theta}{2}O(\sin\theta) & s=0\\
O(\sin\theta) & s=1
\end{cases},\qquad \partial_{\theta}\zeta^{2}_0=\begin{cases}
\cos^2\frac{\theta}{2}O(\sin\theta) & s=0\\
O(\sin\theta) & s=1
\end{cases}.
\end{equation}



\subsection{A priori estimates}
Let $\{\Omega_i\}\subset S^3\setminus\Gamma$ be an exhaustion sequence of bi-axisymmetric domains so that $\Omega_i \subset\Omega_{i+1}$ and $\lim_{i\to\infty}\Omega_i=S^3$. Since the target space is nonpositively curved, standard harmonic map theory \cite{eells1964harmonic} states that there exists a unique solution of the following Dirchlet problem in which the model map is used to prescribe the boundary values
\begin{align}
\begin{split}\label{eq4.1}
\begin{cases}
\tau(\tilde{\Psi}_i)=0\quad \text{in}\quad \Omega_i,\\
\tilde{\Psi}_i=\tilde{\Psi}_0\quad \text{on}\quad \partial\Omega_i.
\end{cases}
\end{split}
\end{align}
We now seek to establish appropriate estimates in order to ensure that the sequence $\tilde{\Psi}_i$ converges on compact subsets of $S^3\setminus\Gamma$. The first step is to achieve an almost uniform distance bound depending only on the distance at a fixed angle $\theta_0\in(0,\pi)$.

\begin{lemma}\label{lemma4.4}
There exists a constant $C$ depending only on $\theta_0$ such that
\begin{equation}\label{eq4.32}
\sup_{\Omega_i}\mathrm{dist}_{N}(\tilde{\Psi}_i,\tilde{\Psi}_0)\leq C
\left[1+\mathrm{dist}_{N}(\tilde{\Psi}_i,\tilde{\Psi}_0)|_{\theta=\theta_0}\right].
\end{equation}
\end{lemma}

\begin{proof}
Let $d_i=\mathrm{dist}_{N}(\tilde{\Psi}_i,\tilde{\Psi}_0)$. The nonpositive curvature of the target manifold ensures (\cite[Lemma 2]{Weinstein}) that on $\Omega_i$
\begin{equation}
\Delta\sqrt{1+d_i^2}\geq-\left(||\tau(\tilde{\Psi}_i)||+||\tau(\tilde{\Psi}_0)||\right)
=-||\tau(\tilde{\Psi}_0)||.
\end{equation}
Solve the boundary value problem
\begin{equation}
\Delta z=||\tau(\tilde{\Psi}_0)||\quad\text{ on }\quad S^{3}_{[0,\theta_0]},\quad\quad\quad
z=0 \quad\text{ on }\quad \partial S^{3}_{[0,\theta_0]},
\end{equation}
where $S^{3}_{[0,\theta_0]}$ is the region of $S^3$ corresponding to the interval $[0,\theta_0]$ in the orbit space $S^3/U(1)^2$. Note that this unique solution is smooth since $\tau(\tilde{\Psi}_0)$ is identically zero in a neighborhood of the poles (and is hence smooth there). It follows that $z+\sqrt{1+d_i^2}$ is subharmonic on $S^{3}_{[\theta_i^-,\theta_0]}$ where $\Omega_i/U(1)^2=(\theta_i^-,\theta_i^+)\subset[0,\pi]$, and so its maximum is achieved on the boundary. Therefore
\begin{equation}
\sup_{S^{3}_{[\theta_i^-,\theta_0]}}\left(z+\sqrt{1+d_i^2}\right)
\leq C(1+d_i(\theta_0)),
\end{equation}
and this leads to the desired bound on $S^{3}_{[\theta_i^-,\theta_0]}$. Analogous arguments yield the bound on $S^{3}_{[\theta_0,\theta_i^+]}$.
\end{proof}

The next goal is to achieve uniform energy and distance bounds, which are based on convexity of the energy resulting from the nonpositive curvature of the target space. Let $F_i:\Omega_i\times[0,1]\to N:=G_{2(2)}/SO(4)$ be a family of geodesic deformations connecting $\tilde{\Psi}_i$ to $\tilde{\Psi}_0$ so that $F_i(\theta,0)=\tilde{\Psi}_i(\theta)$ and $F_i(\theta,1)=\tilde{\Psi}_0(\theta)$, see Figure \ref{fig1}.
The components of $F_i$ will be denoted by
\begin{equation}\label{eq4.3}
 F_i(\theta,t)=(u_i(\theta,t),v_i(\theta,t),w_i(\theta,t),
 \zeta^{1}_i(\theta,t),\zeta^{2}_i(\theta,t),\chi_i(\theta,t),
 \psi^{1}_i(\theta,t),\psi^{2}_i(\theta,t)).
\end{equation}
Since $t\to F_i(\theta,t)$ is a geodesic we have $||\partial_tF_i||=\text{dist}_{N}(\tilde{\Psi}_i,\tilde{\Psi}_0)=d_i$. The second variation of energy yields
\begin{align}\label{eq4.4}
\begin{split}
\frac{d^2}{dt^2}E_{\Omega_i}(F_i)=&\frac{8}{\pi^2}
\int_{\Omega_i}\left[||\nabla^N_{\partial_{\theta}F_i}\partial_tF_i||^2
-{}^N\text{Rm}(\partial_\theta F_i,\partial_tF_i,\partial_tF_i,\partial_\theta F_i)\right] d\mathcal{V}\\
\geq & \frac{8}{\pi^2} \int_{\Omega_i}\left[|\nabla\text{dist}_{N}(\tilde{\Psi}_i,\tilde{\Psi}_0)|^2
-{}^N\text{Rm}(\partial_\theta F_i,\partial_tF_i,\partial_tF_i,\partial_\theta F_i)\right]d\mathcal{V},
\end{split}
\end{align}
where in the second line the Kato inequality
\begin{equation}
||\nabla^N_{\partial_{\theta}F_i}\partial_tF_i||\geq |\nabla ||\partial_tF_i||\text{ }\!|=|\nabla\text{dist}_{N}(\tilde{\Psi}_i,\tilde{\Psi}_0)|
\end{equation}
has been employed.  Since $\tilde{\Psi}_i$ is harmonic the first variation vanishes
at $t=0$, and thus integrating twice produces
\begin{align}\label{eq4.6}
\begin{split}
& E_{\Omega_i}(\tilde{\Psi}_0)\\
\geq& E_{\Omega_i}(\tilde{\Psi}_i)+2\int_{\theta_i^-}^{\theta_i^+}\int_{0}^1\int_{0}^t
\left[|\nabla\text{dist}_{N}(\tilde{\Psi}_i,\tilde{\Psi}_0)|^2
-{}^N\text{Rm}(\partial_\theta F_i,\partial_tF_i,\partial_tF_i,\partial_\theta F_i)\right] d\bar{t} dt \sin\theta d\theta\\
=& E_{\Omega_i}(\tilde{\Psi}_i)
+\int_{\theta_i^-}^{\theta_i^+}\left(|\nabla d_i|^2+f_i d_i^2\right)\sin\theta d\theta,
\end{split}
\end{align}
where $f_i\geq 0$ is given by
\begin{equation}
f_i=-2\int_{0}^{1}\int_{0}^{t}{}^N\text{Rm}\left(\partial_\theta F_i,\frac{\partial_t F_i}{||\partial_t F_i||},\frac{\partial_t F_i}{||\partial_t F_i||},\partial_\theta F_i \right).
\end{equation}

\begin{figure}[h]
\tikzset{middlearrow/.style={
			decoration={markings,
				mark= at position 0.25 with {\arrow{#1}} ,
			},
			postaction={decorate}
		}
	}	
\begin{tikzpicture}[scale=.9, every node/.style={scale=0.8}]
\draw[thick] (-6.5,1.5).. controls (-.5,-.5)and(.5,-.5) .. (6.5,1.5)node[black,font=\large,left=4.7cm,below=1.5cm]{$F_i(\theta,1)=\tilde{\Psi}_0(\theta)$};
\draw[thick] (-4.5,.84)--(-4,1.5).. controls (-3.5,2.5)and(-2.5,1.5) ..(-2,2.5).. controls (-1,3.5)and(1,3.5) .. (2,2.5)node[black,font=\large,right=1cm,above=.2cm]{$F_i(\theta,0)=\tilde{\Psi}_i(\theta)$}.. controls (3,1.5)and(4,2.5)..(4.5,.84);
\draw[thick] (0.3,3.24).. controls (.4,3)and(.4,2.8)..(.2,2.5).. controls (-.2,2.2)and(-.2,1.9)..(0,1.6)node[black,font=\large,right=.1cm]{$F_i(\theta_0,t)$};
\draw[thick,middlearrow={>>}](0,1.6).. controls (.3,1.1) and(.3,.8) ..(-.5,0);
\draw[fill=black] (-4.5,.84) circle [radius=.08] node[black,font=\large,below=.1cm]{$F_i(\theta_i^-,t)$};
\draw[fill=black] (4.5,.84) circle [radius=.08] node[black,font=\large,below=.1cm]{$F_i(\theta_i^+,t)$};
\draw[fill=black] (0.3,3.24) circle [radius=.08];
\draw[fill=black] (-.5,0) circle [radius=.08] node[black,font=\large,below=.1cm]{$x_1:=F_i(\theta_0,1)$};
\end{tikzpicture}
\caption{Geodesic deformation in the target space $N$.}\label{fig1}
\end{figure}
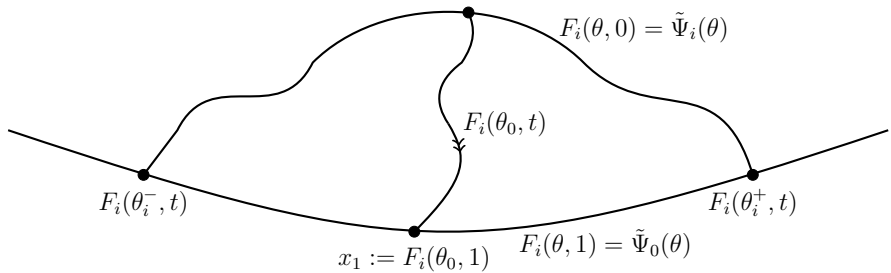

\begin{lemma}\label{lemma67}
There exists a small interval $(\theta_1,\theta_2)\ni \theta_0$ and a uniform constant $\delta>0$ such that
\begin{equation}\label{claim}
\int_{\theta_i^-}^{\theta_i^+}\left(|\nabla d_i|^2+f_i d_i^2\right)\sin\theta d\theta
\geq\delta\int_{\theta_1}^{\theta_2}(|\nabla d_i|^2+d_i^2)d\theta.
\end{equation}
\end{lemma}

\begin{proof}
Let $\gamma_i$ denote the unit speed geodesic connecting $x_1=F_i(\theta_0,1)$ to $F_i(\theta_0,0)$, and let $J$ denote the Jacobi field $\partial_{\theta}F_i$ along this geodesic. In what follows we will suppress the index $i$.
According to the symmetries of the Riemann tensor, the matrix $\langle R(\cdot,\dot{\gamma})\dot{\gamma},\cdot\rangle$ is symmetric and hence admits an orthonormal set of eigenvectors $\{e_j\}$ with eigenvalues $\lambda_j\leq 0$. Since the target manifold is a symmetric space, the Riemann tensor is covariantly constant and therefore $e_j$ are parallel and $\lambda_j$ are constant along the geodesic. We then have
\begin{equation}
f(\theta_0)=-2\int_0^1\int_0^t \langle R(J,\dot{\gamma})\dot{\gamma},J\rangle
=-2\sum \lambda_j \int_0^1 \int_0^t\langle J,e_j\rangle^2.
\end{equation}
The Jacobi equation implies that
\begin{equation}
\partial_t^2\langle J,e_j\rangle=\langle\ddot{J},e_j\rangle
=-\langle R(J,\dot{\gamma})\dot{\gamma},e_j\rangle
=-\lambda_j\langle J,e_j\rangle,
\end{equation}
and therefore
\begin{equation}
\langle J,e_j\rangle(t)=c_{j1} e^{\sqrt{|\lambda_j|}t}+c_{j2} e^{-\sqrt{|\lambda_j|}t}.
\end{equation}
Integrating and completing the square produces
\begin{align}
\begin{split}
f(\theta_0)=&-2\sum \lambda_j\left[\frac{1}{4|\lambda_j|}(e^{2\sqrt{|\lambda_j|}}-1-2\sqrt{|\lambda_j|})
c_{j1}^2+c_{j1}c_{j2}
+\frac{1}{4|\lambda_j|}(e^{-2\sqrt{|\lambda_j|}}-1+2\sqrt{|\lambda_j|})
c_{j2}^2\right]\\
=&\sum \frac{1}{2}(e^{2\sqrt{|\lambda_j|}}-1-2\sqrt{|\lambda_j|})
\left[c_{j1}+\frac{2|\lambda_j|}{(e^{2\sqrt{|\lambda_j|}}-1-2\sqrt{|\lambda_j|})}
c_{j2}\right]^2\\
&+\sum\frac{(e^{-2\sqrt{|\lambda_j|}}-1+2\sqrt{|\lambda_j|})
(e^{2\sqrt{|\lambda_j|}}-1-2\sqrt{|\lambda_j|})-4\lambda_j^2}
{2(e^{2\sqrt{|\lambda_j|}}-1-2\sqrt{|\lambda_j|})}c_{j2}^2.
\end{split}
\end{align}
Note that the coefficient of $c_{j2}^2$ in the last line is positive unless $\lambda_j=0$.

A subsequence of the unit vectors $\dot{\gamma}(1)$ converges, and it may be assumed without loss of generality that this limit is a regular direction by perturbing the point $x_1$ if necessary. Recall that a regular vector is one which lies in a single maximal flat, and that the set of such vectors is dense in the unit sphere. Thus, since the target space $N$ is rank 2 there exists a vector perpendicular to $\dot{\gamma}(1)$, say $e_{1}$, for which the curvature of the resulting 2-plane is bounded away from zero independent of $i$, that is $\lambda_{1}\leq -c<0$. Furthermore, since the constants
$c_{j 1}$ and $c_{j 2}$ are determined by $\langle J,e_j\rangle(1)$ and $\langle\dot{J},e_j\rangle(1)$, and we may choose the model map at $x_1$ to guarantee
that $|\langle J,e_j\rangle(1)|$ stays bounded away from zero independent of $i$, it follows that $f(\theta_0)\geq 2\delta>0$ independent of $i$. This lower bound persists for a small interval around $\theta_0$, and thus yields the desired inequality \eqref{claim}.
\end{proof}

\begin{prop}\label{lemma4.2}
The harmonic energy of the map $\tilde{\Psi}_i$ is uniformly bounded on fixed domains $\Omega\subset S^3\setminus\Gamma$, that is, there exists a constant $C$ independent of $i$ such that
\begin{equation}\label{eq4.2}
E_{\Omega}(\tilde{\Psi}_i)\leq C.
\end{equation}
Moreover the distance function is uniformly bounded
\begin{equation}\label{eq4.32}
\mathrm{dist}_{N}(\tilde{\Psi}_i,\tilde{\Psi}_0)\leq C.
\end{equation}
\end{prop}

\begin{proof}
By equations \eqref{relationIE}, \eqref{eq4.6}, and  $\tilde{\Psi}_i\vert_{\partial\Omega_i}=\tilde{\Psi}_0\vert_{\partial\Omega_i}$ we have
\begin{align}\label{eq4.11}
\begin{split}
\mathcal{I}_{\Omega_i}({\Psi}_0)\geq \mathcal{I}_{\Omega_i}({\Psi}_i) +\int_{\theta_i^-}^{\theta_i^+}\left(|\nabla d_i|^2+f_i d_i^2\right)\sin\theta d\theta.
\end{split}
\end{align}
Observe that
\begin{equation}\label{190}
\mathcal{I}_{\Omega_i}({\Psi}_0)-\mathcal{I}_{\Omega_i}(\Psi_i)
=\tilde{\mathcal{I}}_{\Omega_i}({\Psi}_0)-\tilde{\mathcal{I}}_{\Omega_i}(\Psi_i)
+\int_{\theta_i^-}^{\theta_i^+}\left[12(U_0-U_i)-2s(V_0-V_i)\right]\sin\theta d\theta,
\end{equation}
where $\tilde{\mathcal{I}}_{\Omega_i}({\Psi}_0)$, $\tilde{\mathcal{I}}_{\Omega_i}(\Psi_i)$ are sums of squares and we have used
$\sin\theta \partial_{\theta}h_2=1$. By construction of the model map $\mathcal{I}_{\Omega_i}({\Psi}_0)$ is uniformly bounded, and the integral on the right-hand side of \eqref{190} is controlled by the distance between the model and harmonic maps. It follows that
\begin{equation}\label{191}
C\left(1+\sup_{\Omega_i}d_i\right)\geq\tilde{\mathcal{I}}_{\Omega_i}(\Psi_i)
+\int_{\theta_i^-}^{\theta_i^+}\left(|\nabla d_i|^2+f_i d_i^2\right)\sin\theta d\theta.
\end{equation}
By Lemma \ref{lemma67}
\begin{equation}
\int_{\theta_i^-}^{\theta_i^+}\left(|\nabla d_i|^2+f_i d_i^2\right)\sin\theta d\theta
\geq\delta\int_{\theta_1}^{\theta_2}(|\nabla d_i|^2+d_i^2)d\theta.
\end{equation}
Moreover according to the Sobolev embedding $W^{1,2}\hookrightarrow C^0$ in dimension one, combined with Lemma \ref{lemma4.4}, we have
\begin{equation}
\int_{\theta_1}^{\theta_2}(|\nabla d_i|^2+d_i^2)d\theta
\geq C^{-1}\sup_{S^{3}_{[\theta_1,\theta_2]}} d_i^2
\geq C^{-1}\sup_{\theta=\theta_0} d_i^2
\geq C_{1}^{-1}\sup_{\Omega_i}d_i^2 -C_{2}.
\end{equation}
It follows that
\begin{equation}\label{1111}
C\geq\tilde{\mathcal{I}}_{\Omega_i}(\Psi_i)+C^{-1}\sup_{\Omega_i}d_i^2.
\end{equation}
This immediately gives \eqref{eq4.32}.

In order to obtain \eqref{eq4.2}, use
\begin{equation}
C\geq \tilde{\mathcal{I}}_{\Omega}(\Psi_i)
\end{equation}
from \eqref{1111}.
The pure harmonic energy may take the place of the renormalized energy on the right-hand side. To see this
make the replacements $u_i=U_i+h_1$, $v_i=V_i+h_2$, and $w_i=W_i$ and compute
\begin{align}
\begin{split}
\tilde{\mathcal{I}}_{\Omega}(\Psi_i)=& E_{\Omega}(\tilde{\Psi}_i)
+\int_{\theta^-}^{\theta^+}\left[12(\partial_{\theta}h_1)^2-24\partial_{\theta}h_1
\partial_{\theta}u_i\right]\sin\theta d\theta\\
&+\int_{\theta^-}^{\theta^+}\left[2\partial_{\theta}h_s\partial_{\theta}V_{is}
-(\partial_{\theta}h_s)^2\right]\sin\theta d\theta,
\end{split}
\end{align}
where $\Omega/U(1)^2=(\theta^-,\theta^+)$ and $h_s=2s h_1 +(s-1)h_2$. Observe that
\begin{equation}
\left|\int_{\theta^-}^{\theta^+}24\partial_{\theta}h_1
\partial_{\theta}u_i \sin\theta d\theta\right|\leq 3\int_{\theta^-}^{\theta^+}(\partial_{\theta}u_i)^2\sin\theta d\theta
+48\int_{\theta^-}^{\theta^+}(\partial_{\theta}h_1)^2\sin\theta d\theta
\leq \frac{1}{4}E_{\Omega}(\tilde{\Psi}_i)+C,
\end{equation}
and in a similar way
\begin{equation}
\left|\int_{\theta^-}^{\theta^+}2\partial_{\theta}h_s\partial_{\theta}V_{is}\sin\theta d\theta\right|
\leq \frac{1}{4}E_{\Omega}(\tilde{\Psi}_i)+C.
\end{equation}
Hence
\begin{equation}
\tilde{\mathcal{I}}_{\Omega}(\Psi_i)\geq \frac{1}{2}E_{\Omega}(\tilde{\Psi}_i)-C,
\end{equation}
and the desired result now follows.
\end{proof}

Energy bounds lead to pointwise density bounds with the help of Bochner's formula
and the nonpositive curvature of the target manifold.

\begin{lemma}\label{lemma4.3}
The energy density of the harmonic map $\tilde{\Psi}_i$ is uniformly bounded on compact subsets $\Omega\subset S^3\setminus\Gamma$, that is
\begin{equation}
\sup_{\Omega}|d\tilde{\Psi}_i|\leq C.
\end{equation}
\end{lemma}

\begin{proof}
This is a standard argument \cite{Weinstein}. We include the outline for convenience of the reader. Bochner's formula yields
\begin{equation}\label{eq4.30}
\Delta\left(|d\tilde{\Psi}_i|^2\right)=|\nabla d\tilde{\Psi}_i|^2
+{}^{S^3}\text{Ric}(d\tilde{\Psi}_i,d\tilde{\Psi}_i)
-{}^N\text{Rm}(d\tilde{\Psi}_i,d\tilde{\Psi}_i,d\tilde{\Psi}_i,d\tilde{\Psi}_i)\geq 0.
\end{equation}
The squared density is then subharmonic, and the De Giorgi-Nash-Moser inequality combined with Proposition \ref{lemma4.2} produces
\begin{equation}\label{eq4.31}
\sup_{\Omega}|d\tilde{\Psi_i}|^2\leq C' E_{\Omega'}(\tilde{\Psi}_i)\leq C,
\end{equation}
where $C$ is independent of $i$ and $\Omega\subset\Omega'$.
\end{proof}

\subsection{Existence and uniqueness}

We say that a map $\tilde{\Psi}:S^3\setminus\Gamma\rightarrow N$ is \textit{asymptotic} to the model map $\tilde{\Psi}_0$ if they remain within a bounded distance from one another even on approach to the poles, that is $\mathrm{dist}_{N}(\tilde{\Psi},\tilde{\Psi}_0)\leq C$. A map which is asymptotic to the model map possesses the same singular behavior as the model map near the poles.

\begin{theorem}\label{theorem4.5}
There exists a harmonic map $\tilde{\Psi}:S^3\setminus\Gamma\rightarrow G_{2(2)}/SO(4)$ which is asymptotic to the model map $\tilde{\Psi}_0$.
\end{theorem}

\begin{proof}
The harmonic map equations satisfied by the sequence $\tilde{\Psi}_i$,
combined with the pointwise gradient bound (Lemma \ref{lemma4.3}) and $L^\infty$
bound \eqref{eq4.32}, imply uniform a priori estimates for all derivatives on fixed domains $\Omega\subset S^3\setminus\Gamma$. In the usual way, by choosing a sequence of exhausting domains and taking a diagonal subsequence, we find a sequence of maps $\tilde{\Psi}_{i_l}$ which converges on compact subsets to a smooth harmonic map $\tilde{\Psi}$. The limit also satisfies the $L^\infty$ bound and is thus asymptotic to the model map.
\end{proof}

In order to establish uniqueness for harmonic maps asymptotic to the same model map, we will need the following preliminary result.

\begin{lemma}\label{lemma4.6}
Suppose that $\tilde{\Psi}_1$ and $\tilde{\Psi}_2$ are two harmonic maps from $S^3\setminus\Gamma\rightarrow N$ such that $\mathrm{dist}_{N}(\tilde{\Psi}_1,\tilde{\Psi}_2)$ is a nonzero constant. Let $S\subset N$ be the 2-dimensional submanifold generated by the geodesic deformation $F(\theta,t)$ connecting
$\tilde{\Psi}_1$ to $\tilde{\Psi}_2$.
If the sectional curvature of the coordinate 2-planes $\mathcal{K}(\partial_t F,\partial_\theta F)$ vanishes, then $S$ is totally geodesic and flat.
\end{lemma}

\begin{proof}
We first show that $S$ is flat. Consider the Gauss equations
\begin{equation}\label{eq4.54}
\begin{split}
{}^N\text{Rm}(\partial_\theta F,\partial_t F,\partial_t F,\partial_\theta F)=&{}^S\text{Rm}(\partial_\theta F,\partial_t F,\partial_t F,\partial_\theta F)\\
& -\langle A(\partial_t F,\partial_t F),A(\partial_\theta F,\partial_\theta F)\rangle+||A(\partial_t F,\partial_\theta F)||^2,
\end{split}
\end{equation}
where ${}^N\text{Rm}$ and ${}^S\text{Rm}$ are the Riemann curvature tensors of $N$ and $S$, respectively, and $A$ is the second fundamental form.
According to the assumption on the sectional curvature of the coordinate 2-planes, it suffices to show that the terms involving $A$ vanish. Let $\mathbf{n}$ be a unit normal vector on $S$, then the definition of the second fundamental form together with the fact that $t\to F(\theta,t)$ is a geodesic produces
\begin{equation}\label{eq4.55}
A_{\mathbf{n}}(\partial_t F,\partial_t F)=\langle \nabla_{t}\mathbf{n}, \partial_t F\rangle=-\langle\mathbf{n}, \nabla_{t}\partial_t F\rangle=0.
\end{equation}
Moreover by assumption $||\partial_t F||=\mathrm{dist}_{N}(\tilde{\Psi}_1,\tilde{\Psi}_2)=\mathrm{const}$, and thus
\begin{equation}\label{eq4.56}
0=\frac{1}{2}\partial_{\theta} ||\partial_t F||^2=\langle\nabla_{\theta}\partial_t F, \partial_t F\rangle
=\partial_t \langle\partial_\theta F, \partial_t F\rangle-\langle\partial_\theta F, \nabla_{t}\partial_t F\rangle
=\partial_t \langle\partial_\theta F, \partial_t F\rangle.
\end{equation}
It follows that
\begin{equation}\label{eq4.57}
\langle\partial_\theta F, \partial_t F\rangle(\theta,1)=\langle\partial_\theta F, \partial_t F\rangle(\theta,0).
\end{equation}
Next, since $\mathcal{K}(\partial_t F,\partial_\theta F)=0$ we have
\begin{equation}\label{eq4.58}
0={}^N\text{Rm}(\partial_t F,\partial_\theta F,\partial_t F,\partial_\theta F)
=\langle\nabla_{t}\nabla_{\theta}\partial_t F-\nabla_{\theta}\nabla_{t}\partial_t F,\partial_{\theta}F\rangle
=\partial_t\langle\nabla_{\theta}\partial_t F,\partial_{\theta}F\rangle-||\nabla_{\theta }\partial_t F||^2.
\end{equation}
Integrating over $t$ and using \eqref{eq4.57} yields
\begin{align}\label{eq4.59}
\begin{split}
\int_{0}^1||\nabla_{\theta }\partial_t F||^2 dt=&\langle\nabla_{\theta}\partial_t F,\partial_{\theta}F\rangle(\theta,1)-\langle\nabla_{\theta}\partial_t F,\partial_{\theta}F\rangle(\theta,0)\\
=&\partial_{\theta}\left(\langle\partial_t F,\partial_\theta F\rangle(\theta,1)
-\langle\partial_t F,\partial_\theta F\rangle(\theta,0)\right)\\
&-\langle\partial_t F,\nabla_\theta \partial_\theta F\rangle (\theta,1)
+\langle\partial_t F,\nabla_\theta \partial_\theta F\rangle (\theta,1)\\
=&\cot\theta \left(\langle\partial_t F,\partial_\theta F\rangle (\theta,1)
-\langle\partial_t F,\partial_\theta F\rangle (\theta,0)\right)\\
=&0.
\end{split}
\end{align}
In the above computation we have employed the fact that as a result of the harmonic map equations $\theta\rightarrow F(\theta,1)$ and $\theta\rightarrow F(\theta,0)$ are pre-geodesics, that is they fail to be geodesics only due to their parameterization and satisfy $\nabla_\theta \partial_\theta F=-\cot\theta \partial_\theta F$.
We now have
\begin{equation}\label{eq4.60}
A_{\mathbf{n}}(\partial_t F,\partial_\theta F)=\langle\nabla_{t}\mathbf{n}, \partial_\theta F\rangle=-\langle\mathbf{n}, \nabla_{t}\partial_\theta F\rangle=0,
\end{equation}
and therefore $S$ is flat.

To show that $S$ is totally geodesic it remains to demonstrate that
\begin{equation}\label{eq4.63}
0=A_{\mathbf{n}}(\partial_\theta F,\partial_\theta F)=\langle\nabla_{\theta}\mathbf{n}, \partial_\theta F\rangle=-\langle\mathbf{n}, \nabla_{\theta}\partial_\theta F\rangle.
\end{equation}
By differentiating $\nabla_{t }\partial_\theta F=0$ with respect to $\theta$ we find
\begin{equation}\label{eq4.65}
0=\nabla_{\theta }\nabla_{t}\partial_\theta F=\nabla_{t}\nabla_{\theta}\partial_\theta F+{}^N R(\partial_{\theta}F,\partial_t F)\partial_{\theta}F.
\end{equation}
Since the curvature tensor is covariantly constant in a symmetric space, it follows that
\begin{equation}
\nabla_{t}\nabla_{t}\nabla_{\theta}\partial_\theta F=0.
\end{equation}
Let now $\mathbf{e}(t)$ be a parallel transported vector field along the geodesic $t\rightarrow F(\theta,t)$ which is normal to $S$, then
\begin{equation}
\partial_{t}^2 \langle \nabla_{\theta}\partial_{\theta}F,\mathbf{e}\rangle=0.
\end{equation}
Furthermore the harmonic map equations show that
\begin{equation}
\langle \nabla_{\theta}\partial_{\theta}F,\mathbf{e}\rangle(\theta,1)
=\langle \nabla_{\theta}\partial_{\theta}F,\mathbf{e}\rangle(\theta,0)=0,
\end{equation}
and hence $\langle \nabla_{\theta}\partial_{\theta}F,\mathbf{e}\rangle(\theta,t)=0$ for all $t$. As $\mathbf{e}$ was arbitrarily chosen normal to $S$, it follows that \eqref{eq4.63} holds.
\end{proof}

We are now in a position to state the basic uniqueness result for the singular harmonic maps having the same asymptotics.

\begin{theorem}\label{theorem4.7}
Suppose that $\tilde{\Psi}_1$ and $\tilde{\Psi}_2$ are two harmonic maps from $S^3\setminus\Gamma\rightarrow N$ which are asymptotic to each other, that is their mutual distance $\mathrm{dist}_{N}(\tilde{\Psi}_1,\tilde{\Psi}_2)$ remains bounded.
Then there exists an isometry of the target space $\varphi:N\rightarrow N$ such that $\tilde{\Psi}_2(s)=\varphi\circ\tilde{\Psi}_1(s+c)$ where $s$ denotes arc-length parameter and $c$ is a constant.
\end{theorem}

\begin{proof}
As before let $F(\theta,t)$ denote a geodesic deformation connecting $\tilde{\Psi}_1$ to $\tilde{\Psi}_2$. Then the Poincar\'{e} inequality, equation \eqref{eq4.11}, and
\cite[Theorem 7.1]{Alaee:2017pex} produce
\begin{align}\label{eq4.74}
\begin{split}
\mathcal{I}({\Psi}_2)\geq& \mathcal{I}({\Psi}_1) +C\int_{0}^{\pi}\left(\mathrm{dist}_{N}(\tilde{\Psi}_1,\tilde{\Psi}_2)-\mathbf{D}\right)^2\sin\theta d\theta\\
&+2\int_{0}^{\pi}\int_{0}^1\int_0^t |\mathcal{K}(\partial_t F,\partial_\theta F)|\left(||\partial_t F||^2||\partial_\theta F||^2-\langle\partial_t F,\partial_\theta F\rangle\right) \sin\theta d\bar{t}dtd\theta,
\end{split}
\end{align}
where $\mathbf{D}$ is a constant that represents the average value of $\text{dist}_{N}(\tilde{\Psi}_1,\tilde{\Psi}_2)$. Since both $\tilde{\Psi}_1$ and $\tilde{\Psi}_2$ are harmonic, their roles may be reversed in the above inequality.
It follows that
\begin{equation}\label{eq4.77}
\mathcal{I}({\Psi}_1)=\mathcal{I}({\Psi}_2),
\end{equation}
\begin{equation}\label{eq4.78}
\mathrm{dist}_{N}(\tilde{\Psi}_1,\tilde{\Psi}_2)=\mathbf{D},
\end{equation}
and
\begin{equation}\label{eq4.79}
|\mathcal{K}(\partial_t F,\partial_\theta F)|\left(||\partial_t F||^2||\partial_\theta F||^2-\langle\partial_t F,\partial_\theta F\rangle\right)=0.
\end{equation}

\textbullet\,\,{\bf Case I: $\mathcal{K}(\partial_t F,\partial_\theta F)\neq 0$ at some point.}
If $\mathcal{K}(\partial_t F,\partial_\theta F)\neq 0$ at $(\theta_0,t_0)$, then by continuity this persists for all nearby $(\theta,t)$. Thus by \eqref{eq4.79}, there exists a neighborhood of $(\theta_0,t_0)$ on which
\begin{equation}\label{eq4.80}
||\partial_t F||^2||\partial_\theta F||^2=\langle\partial_t F,\partial_\theta F\rangle.
\end{equation}
The Cauchy-Schwarz inequality then implies that these two vectors are multiples of each other, that is
\begin{equation}\label{eq4.81}
 a\partial_t F=\partial_\theta F.
\end{equation}
Since $J:=\partial_{\theta}F$ is a Jacobi field, it is determined by $J(t_0)$ and $\partial_t J(t_0)$. But for $(\theta,t)$ close to $(\theta_0,t_0)$
\begin{equation}\label{eq4.82}
J(t)=\partial_{\theta}F(\theta,t)=a(\theta,t)\partial_{t}F(\theta,t),
\end{equation}
and so
\begin{equation}\label{eq4.83}
J(t_0)=a(\theta,t_0)\partial_{t}F(\theta,t_0),\quad\quad
\partial_t J(t_0)=\partial_{t}a(\theta,t_0)\partial_{t}F(\theta,t_0).
\end{equation}
The Jacobi equation then yields
\begin{equation}
\partial_\theta F=J(t)=\left[a_1(\theta)(t-t_0)+a_2(\theta)\right]\partial_t F.
\end{equation}
In particular, the pre-geodesics $\theta\rightarrow F(\theta,1)$ and $\theta\rightarrow F(\theta,0)$ coincide with the geodesic $t\rightarrow F(\theta_0,t)$ up to reparameterization. We then have
$\tilde{\Psi}_1(s)=\tilde{\Psi}_2(s+c)$ for some constant $c$, where $s$ denotes  arc-length parameter.

\textbullet\,\,{\bf Case II: $\mathcal{K}(\partial_t F,\partial_\theta F)= 0$ for all points.}
If $\mathbf{D}=0$ then we are done, so assume that $\mathbf{D}\neq 0$.
The subset $S\subset N$ generated by $F(\theta,t)$ is then a 2-dimensional submanifold. According to the assumptions of this case, Lemma \ref{lemma4.6} implies that $S$ is totally geodesic and flat. Let $p\in S$ and $O\subset T_p N$ be the 2-plane spanned by $\partial_\theta F$ and $\partial_{t}F$.
The surface $S$ may be extended to be a complete maximal flat by setting $S=\mathrm{exp}_p(O)$. The pre-geodesics $\tilde{\Psi}_1(\theta)=F(\theta,1)$ and $\tilde{\Psi}_2(\theta)=F(\theta,0)$ are parallel straight lines in $S$ in light of \eqref{eq4.78}. Furthermore, the Iwasawa decomposition of the isometry group $G_{2(2)}$
may be used to show that there is a subgroup which operates transitively on the maximal flat $S$, see \cite[Section 6]{Khuri:2017xsc} for details in the $SL(3,\mathbb{R})/SO(3)$ setting which carries over without change to the present situation as both target spaces are rank 2. Thus, there exists an isometry
of the target space $\varphi:N\rightarrow N$ which maps $\tilde{\Psi}_1$ onto $\tilde{\Psi}_2$ up to translation in the arclength parameter.
\end{proof}

\subsection{Proof of the main theorem}

In Section 3 it was demonstrated that bi-axisymmetric near-horizon geometry solutions of 5-dimensional minimal supergravity correspond to singular harmonic maps from $S^3\setminus\Gamma\rightarrow G_{2(2)}/SO(4)$. These harmonic maps were shown to exist in Section 4 asymptotic to a given model map. The choice of model map gives rise to corresponding minimal supergravity near-horizon geometries having the prescribed horizon topology $S^3$, $S^1\times S^2$, or $L(p,q)$, and with prescribed electric charge $\mathcal{Q}$, angular momenta $\mathcal{J}_i$, and dipole charge $\mathcal{D}$ (in the ring case). The uniqueness statement of Theorem \ref{maintheorem} follows directly from Theorem \ref{theorem4.7}.


\bibliographystyle{abbrv}
\bibliography{masterfile}

\begin{thebibliography}{10}

\bibitem{Alaee:2015pwa}
A.~Alaee, M.~Khuri, and H.~Kunduri.
\newblock {Proof of the mass-angular momentum inequality for bi-axisymmetric
  black holes with spherical topology}.
\newblock {\em Adv. Theor. Math. Phys.}, 20(6):1397--1441, 2016.

\bibitem{Alaee:2017ygv}
A.~Alaee, M.~Khuri, and H.~Kunduri.
\newblock {Mass-angular momentum inequality for black ring spacetimes}.
\newblock {\em Phys. Rev. Lett.}, 119(7):071101, 2017.

\bibitem{Alaee:2016jlp}
A.~Alaee, M.~Khuri, and H.~Kunduri.
\newblock {Relating mass to angular momentum and charge in 5-dimensional
  minimal supergravity}.
\newblock {\em Ann. Henri Poincare}, 18(5):1703--1753, 2017.

\bibitem{Alaee:2017pex}
A.~Alaee, M.~Khuri, and H.~Kunduri.
\newblock {Bounding horizon area by angular momentum, charge, and cosmological
  constant in 5-dimensional minimal supergravity}.
\newblock {\em Ann. Henri Poincare}, 20(2):481--525, 2019.

\bibitem{Bouchareb:2007ax}
A.~Bouchareb, G.~Clement, C.-M. Chen, D.~Gal'tsov, N.~Scherbluk, and T.~Wolf.
\newblock {G2 generating technique for minimal D=5 supergravity and black
  rings}.
\newblock {\em Phys. Rev. D}, 76:104032, 2007.
\newblock [Erratum: Phys. Rev. D 78, 029901 (2008)].

\bibitem{chrusciel2012stationary}
P.~Chrusciel, J.~Costa, and M.~Heusler.
\newblock Stationary black holes: uniqueness and beyond.
\newblock {\em Living Reviews in Relativity}, 15(7), 2012.

\bibitem{dain2012geometric}
S.~Dain.
\newblock Geometric inequalities for axially symmetric black holes.
\newblock {\em Classical and Quantum Gravity}, 29(7):073001, 2012.

\bibitem{Dunajski:2016rtx}
M.~Dunajski, J.~Gutowski, and W.~Sabra.
\newblock {Einstein-Weyl spaces and near-horizon geometry}.
\newblock {\em Class. Quant. Grav.}, 34(4):045009, 2017.

\bibitem{eells1964harmonic}
J.~Eells and J.~Sampson.
\newblock Harmonic mappings of riemannian manifolds.
\newblock {\em American Journal of Mathematics}, 86(1):109--160, 1964.

\bibitem{emparan2008black}
R.~Emparan and H.~Reall.
\newblock Black holes in higher dimensions.
\newblock {\em Living Reviews in Relativity}, 11(6):0801--3471, 2008.

\bibitem{Figueras2010}
P.~Figueras and J.~Lucietti.
\newblock On the uniqueness of extremal vacuum black holes.
\newblock {\em Classical and Quantum Gravity}, 27(9):095001, 2010.

\bibitem{galloway2006rigidity}
G.~J. Galloway.
\newblock Rigidity of marginally trapped surfaces and the topology of black
  holes.
\newblock {\em Communications in Analysis and Geometry}, 16(1):217--229, 2008.

\bibitem{Galloway2006}
G.~J. Galloway and R.~Schoen.
\newblock A generalization of hawking's black hole topology theorem to higher
  dimensions.
\newblock {\em Commun. Math. Phys.}, 266(2):571--576, 2006.

\bibitem{hollands2009stationary}
S.~Hollands and A.~Ishibashi.
\newblock On the `stationary implies axisymmetric' theorem for extremal black
  holes in higher dimensions.
\newblock {\em Commun. Math. Phys.}, 291(2):443--471, 2009.

\bibitem{hollands2010all}
S.~Hollands and A.~Ishibashi.
\newblock All vacuum near horizon geometries in $d$-dimensions with $(d-3)$
  commuting rotational symmetries.
\newblock 10(8):1537--1557, 2010.

\bibitem{hollands2012black}
S.~Hollands and A.~Ishibashi.
\newblock Black hole uniqueness theorems in higher dimensional spacetimes.
\newblock {\em Classical and Quantum Gravity}, 29(16):163001, 2012.

\bibitem{Hollands2007}
S.~Hollands, A.~Ishibashi, and R.~Wald.
\newblock A higher dimensional stationary rotating black hole must be
  axisymmetric.
\newblock {\em Commun. Math. Phys.}, 271(3):699--722, 2007.

\bibitem{Hollands2008}
S.~Hollands and S.~Yazadjiev.
\newblock Uniqueness theorem for 5-dimensional black holes with two axial
  killing fields.
\newblock {\em Commun. Math. Phys.}, 283(3):749--768, 2008.

\bibitem{moncrief2008symmetries}
J.~Isenberg and V.~Moncrief.
\newblock Symmetries of higher dimensional black holes.
\newblock {\em Classical and Quantum Gravity}, 25(19):195015, 2008.

\bibitem{KhuriWeinsteinYamada1}
M.~Khuri, G.~Weinstein, and S.~Yamada.
\newblock {Asymptotically locally Euclidean/Kaluza-Klein stationary vacuum
  black holes in 5 dimensions}.
\newblock {\em PTEP. Prog. Theor. Exp. Phys.}, (5):053E01, 2018.

\bibitem{Khuri:2017xsc}
M.~Khuri, G.~Weinstein, and S.~Yamada.
\newblock {Stationary vacuum black holes in 5 dimensions}.
\newblock {\em Comm. Partial Differential Equations}, 43(8):1205--1241, 2018.

\bibitem{Kunduri2009}
H.~Kunduri and J.~Lucietti.
\newblock A classification of near-horizon geometries of extremal vacuum black
  holes.
\newblock {\em Journal of Mathematical Physics}, 50(8):082502, 2009.

\bibitem{Kunduri:2009ud}
H.~Kunduri and J.~Lucietti.
\newblock {Static near-horizon geometries in five dimensions}.
\newblock {\em Classical and Quantum Gravity}, 26(24):245010, 2009.

\bibitem{kunduri2011constructing}
H.~Kunduri and J.~Lucietti.
\newblock Constructing near-horizon geometries in supergravities with hidden
  symmetry.
\newblock {\em Journal of High Energy Physics}, 2011(7):1--31, 2011.

\bibitem{Kunduri:2013ana}
H.~Kunduri and J.~Lucietti.
\newblock {Classification of near-horizon geometries of extremal black holes}.
\newblock {\em Living Reviews in Relativity}, 16:8, 2013.

\bibitem{kunduri2014supersymmetric}
H.~Kunduri and J.~Lucietti.
\newblock Supersymmetric black holes with lens space topology.
\newblock {\em Physical Review letters}, 113(21):211101, 2014.

\bibitem{Kunduri:2007vf}
H.~Kunduri, J.~Lucietti, and H.~Reall.
\newblock {Near-horizon symmetries of extremal black holes}.
\newblock {\em Classical and Quantum Gravity}, 24:4169--4190, 2007.

\bibitem{Li:2015wsa}
C.~Li and J.~Lucietti.
\newblock {Transverse deformations of extreme horizons}.
\newblock {\em Classical and Quantum Gravity}, 33(7):075015, 2016.

\bibitem{Li:2018knr}
C.~Li and J.~Lucietti.
\newblock {Electrovacuum spacetime near an extreme horizon}.
\newblock 2018.

\bibitem{Myers1986}
R.~C. Myers and M.~Perry.
\newblock Black holes in higher dimensional space-times.
\newblock {\em Annals of Physics}, 172(2):304--347, 1986.

\bibitem{Pomeransky2006}
A.~Pomeransky and R.~Sen'kov.
\newblock Black ring with two angular momenta.
\newblock {\em arXiv preprint hep-th/0612005}, 2006.

\bibitem{Possel:2003yw}
M.~Possel and S.~Silva.
\newblock {Hidden symmetries in minimal five-dimensional supergravity}.
\newblock {\em Phys. Lett. B}, 580(3-4):273--279, 2004.

\bibitem{Reall:2002bh}
H.~S. Reall.
\newblock {Higher dimensional black holes and supersymmetry}.
\newblock {\em Phys. Rev.}, D68:024024, 2003.
\newblock [Erratum: Phys. Rev.D70,089902(2004)].

\bibitem{Strominger1996}
A.~Strominger and C.~Vafa.
\newblock Microscopic origin of the bekenstein-hawking entropy.
\newblock {\em Physics Letters B}, 379(1):99--104, 1996.

\bibitem{tomizawa2016supersymmetric}
S.~Tomizawa and M.~Nozawa.
\newblock Supersymmetric black lenses in five dimensions.
\newblock {\em Phys. Rev. D}, 94(4):044037, 2016.

\bibitem{Weinstein}
G.~Weinstein.
\newblock Harmonic maps with prescribed singularities into hadamard manifolds.
\newblock {\em Mathematical Research Letters}, 3(6):835--844, 1996.

\end{thebibliography}

\end{document}